\documentclass[11pt,leqno,textwidth=10cm]{amsart}
\usepackage{amssymb}

\usepackage{mathrsfs}
\usepackage{mathtools}
\usepackage{abstract}

\usepackage{etoolbox}
\makeatletter
\patchcmd{\@maketitle}{\newpage}{}{}{} 
\makeatother
\numberwithin{equation}{section}

\newtheoremstyle{fancy}{}{}{\itshape}{}{\textsc\bgroup}{.\egroup}{ }{}
\newtheoremstyle{fancy2}{}{}{\rm}{}{\textsc\bgroup}{.\egroup}{ }{}
\theoremstyle{fancy}
\newcounter{intro}

\numberwithin{equation}{section}    
\newtheorem{cor}[intro]{Corollary}
\newtheorem{lem}[intro]{Lemma}
\newtheorem{prop}[intro]{Proposition}
\newtheorem{thm}[intro]{Theorem}

\newtheorem{named}{\name}
\newcommand{\name}{Proof of}

\theoremstyle{fancy2}
\newtheorem{dfn}[intro]{Definition}
\newcounter{axa}

\newtheorem{rem}[intro]{Remark}

\newcommand{\cref}[1]{Corollary~\ref{#1}}

\renewcommand{\div}{\operatorname{div}}

\newcommand{\supp}{\operatorname{supp}}

\newcommand{\vol}{\operatorname{vol}}

\newcommand{\riem}{\mathbf{Riem}}

\newcommand{\Chrt}[3]{\tilde{\Gamma}^{#1}_{#2 #3}}

\newcommand{\Chrtg}[3]{\tilde{\Gamma}^{#1}_{#2 #3}(g)}
\newcommand{\Chr}[3]{\Gamma^{#1}_{#2 #3}}
\newcommand{\hChr}[3]{\widehat{\Gamma}^{#1}_{#2 #3}}

\newcommand{\Chrn}[3]{\boldsymbol{\Gamma}^{#1}_{#2 #3}}
\newcommand{\Chrtn}[4]{\/^{(#4)}\tilde{\Gamma}^{#1}_{#2 #3}}

\newcounter{subequation} 
 \newlength\mtabskip\mtabskip=-1.25cm
   
	 \def\mtabLong{long} 
 

\newcommand{\eq}[1]{\begin{equation}#1\end{equation}}

\newcommand{\alg}[1]{\begin{aligned}#1\end{aligned}}

\newcommand{\bc}{\begin{cases}}
\newcommand{\ec}{\end{cases}}

\newcommand{\p}[1]{\partial_{#1}}

\newcommand{\pp}[1]{\partial_{p^{#1}}}
\newcommand{\hp}[1]{\widehat{\partial}_{#1}}

\newcommand{\Ab}[1]{\|#1\|}
\newcommand{\ab}[1]{|#1|}

\newcommand{\Abi}[1]{\|#1\|_{\infty}}

\newcommand{\Abk}[2]{\|#1\|_{H^{#2}}}

\newcommand{\abg}[1]{|#1|_g}

\newcommand{\al}{\alpha}
\newcommand{\be}{\beta}
\newcommand{\ga}{\gamma}

\newcommand{\la}{\lambda}
\newcommand{\Si}{\Sigma}

\newcommand{\De}{\Delta}

\newcommand{\na}{\nabla}

\newcommand{\Ip}[1]{\int #1 dp}

\newcommand{\sg}{\sqrt{g}}

\newcommand{\br}[1]{\boldsymbol{\varrho}_{#1}}

\newcommand{\ev}[2]{\mathscr E_{#1,#2} }

\newcommand{\mcl}[1]{\mathcal{#1}}
\newcommand{\mcr}[1]{\mathscr{#1}}

\newcommand{\mbf}[1]{\mathbf{#1}}

\newcommand{\A}{\mbf{A}}
\newcommand{\B}{\mbf{B}}

\newcommand{\et}{\mbf E_{\mathrm{tot}}}

\def\div{ \mbox{div}}

\begin{document}


\title[Nonlinear stability of the Milne model with matter]{ Nonlinear stability of the Milne model with matter}
\author[L.~Andersson, D.~Fajman]{ Lars Andersson, David Fajman}


\date{\today}


\address{
Lars Andersson\newline
Max-Planck Institute for Gravitational Physics\newline
Am M\"uhlenberg 1\newline
D-14476 Potsdam, Germany\newline
laan@aei.mpg.de
}

\address{
David Fajman\newline
Faculty of Physics, 
University of Vienna\\ \newline
Boltzmanngasse 5\\ \newline
1090 Vienna, Austria\\ \newline
David.Fajman@univie.ac.at
}

\maketitle
\begin{abstract}
We show that any 3+1-dimensional Milne model is future nonlinearly, asymptotically stable in the set of solutions to the Einstein-Vlasov system. For the analysis of the Einstein equations we use the constant-mean-curvature-spatial-harmonic gauge. For the distribution function the proof makes use of geometric $L^2$-estimates based on the Sasaki-metric. The resulting estimates on the energy momentum tensor are then upgraded by employing the natural continuity equation for the energy density. 
\end{abstract}


\section{Introduction}

\subsection{Cosmological spacetimes and stability}
We consider the following class of cosmological vacuum spacetimes. Let the $M$ be a closed 3-manifold admitting an Einstein metric $\ga$ with negative Einstein constant $\mu=-\frac{2}{9}$, i.e.
\eq{
R_{ij}[\ga]=-\frac29\ga_{ij},
}
where the specific value of $\mu$ is chosen for convenience. A spacetime of the form $((0,\infty)\times M,\overline g)$ with
\eq{\label{background}
\overline g=-dt^2+\frac{t^2}9 \cdot \gamma
}
is known as a \emph{Milne model} and is a solution to the vacuum Einstein equations. Its future nonlinear stability under the vacuum Einstein flow has been shown in \cite{AnMo11}  and constitutes the second stability result for the vacuum Einstein equations without symmetry assumptions beside the corresponding one for Minkowski space \cite{CK92}. While the stability of the Minkowski spacetime under the vacuum Einstein flow has been generalized to several Einstein-matter systems \cite{BZ,LR,LM,Ta16,FJS17-2,LT17} this is not the case for the Milne model. We address this problem for the Einstein-Vlasov system.

\subsection{The stability problem for the Einstein-Vlasov system}
The Einstein-Vlasov system (EVS) reads
\eq{\alg{
\overline R_{\mu\nu}-\frac12\overline R \overline g_{\mu\nu}&= \int_{\mathscr P_x} f p_{\mu}p_\nu \mu_{\mathscr P_x}\\
X_{\overline g}f&=0,
}}
where $X_{\overline g}$ denotes the geodesic spray and $f$ a distribution function with domain $\mathscr P_x\subset T\overline M$, the mass-shell of future directed particles for a fixed mass $m$. It models spacetimes containing ensembles of self-gravitating, collisionless particles and constitutes an almost accurate model for spacetime on large scales, where collisions are negligible and galaxies and galaxy clusters indeed interact solely be their mutual self-gravitation. Its mathematical study in the context of the Cauchy problem dates back to the first works by Rein and Rendall on the evolution of spherically-symmetric perturbations of Minkowski space \cite{ReRe92} and the construction of static nonvacuum solutions \cite{ReRe93}. Substantial progress in the study of the EVS happened since then. For a complete overview we refer to the review article by Andr\'easson \cite{An11}. Regarding the nonlinear stability problem, in particular without symmetry assumptions, first results have appeared recently
considering different geometric scenarios. Ringstr\"om's monumental work, which in particular contains a detailed local-existence theory, addresses the stability problem for exponentially expanding cosmological models \cite{Ri13}. These correspond to the presence of a positive cosmological constant in the Einstein equations, which in his case is realized by a scalar field with suitable potential. This has later been extended by Andr\'easson and Ringstr\"om to prove stability of $T^3$ Gowdy symmetric solutions (in the class of all solutions without symmetry assumptions) \cite{AnRi13}. Furthermore, the stability of Minkowski space for the Einstein-Vlasov system for massless particles has been proven by Taylor \cite{Ta16}. The stability of 2+1-dimensional cosmological spacetimes for the Einstein-Vlasov system has been proven by the second author \cite{Fa16, Fa16-2}.
We remark that in the physically interesting case of 3+1 dimensions, nonlinear stability results until very recently either required a positive cosmological constant or a restriction to the massless case. A recent series of works then established the stability of Minkowski space for the Einstein-Vlasov system by a vector-field-method approach \cite{FJS15,FJS17, FJS17-2} and also independently \cite{LT17}.\\
In the present paper is we establish the first stability result for the Einstein-Vlasov system in 3+1 dimensions in the cosmological case with vanishing cosmological constant. Moreover, to our knowledge, the present work presents the first stability result to an Einstein-matter system with vanishing cosmological constant in the cosmological case.\\ 
Further stability results for cosmological spacetimes with matter models exist but to our knowledge consider the case of a positive cosmological constant. We refer here to the works of Rodnianski-Speck and Speck on the Einstein-Euler system \cite{RoSp13,S12}, Had$\check{\mathrm{z}}$i$\acute{\mathrm{c}}$-Speck on the Einstein-dust system \cite{HS15}, Friedrich on the Einstein-dust system \cite{Fr17} and Olyniyk on the Einstein-fluid system \cite{Ol16}. 

\subsection{Nonvacuum stability of the Milne model -- Main theorem}
To prove nonlinear stability of any Milne model within the class of solutions to the Einstein-Vlasov system we 
first extend the rescaling of the geometry by the mean curvature function as done in \cite{AnMo11} to the nonvacuum case by rescaling the momentum variables $\tilde p$ accordingly. The choice of rescaling here is motivated by the behavior of the momentum support for solutions to the transport equation on the background \eqref{background}, which decreases as $\tilde p\approx t^{-2}$. The mass-shell relation of massive particles, however, prevents from obtaining a system of autonomous equations, as it occurs for the vacuum system. In the present case, some explicit time functions remain in the rescaled equations, which appear in conjunction with the energy-momentum tensor.
We then combine the technique of corrected energies to control the perturbation of the geometry as developed for the vacuum case in \cite{AnMo11} with the technique of $L^2$-Sobolev-energies for the distribution function based on the Sasaki metric on the spatial tangent bundle derived in \cite{Fa16}.
\subsubsection{Rescaling}
As in \cite{AnMo11} we use here a rescaling of the geometric variables (and in addition of the matter quantities) in terms of the mean curvature $\tau$. This rescaling is introduced in \eqref{rescaling}. Moreover, a logarithmic time variable $T$ is then introduced in \eqref{log-time}. The following discussion and statement of the main theorem is conducted with respect to these variables.
\subsubsection{Difficulties in 3+1 dimensions}
 A fundamental difference to the 2+1-dimensional case considered in \cite{Fa16} is the different structure of the matter quantity appearing in the elliptic equation for the lapse function (cf.~$\tau\eta$ in equation \eqref{lapse}). In dimension 3+1, after the appropriate rescaling, we find that this quantity does not decay faster than $e^{-T}$. This occurs already on the level of the unperturbed background geometry and implies that Sobolev norms of the gradient of the lapse function only decay as $e^{-T}$. In view of this slow decay a critical problem arises when the $L^2$-Sobolev estimates for the distribution function are considered. In the transport equation the critical term reads \eq{e^{T}\cdot N\nabla^aNp^0\frac{\partial f}{\partial p^a},}
  written in rescaled variables. Roughly analyzed\footnote{For details we refer to the $L^2$-estimates for the distribution function, which immediately clarify this conclusion.}, the decay of the lapse, of the form $\nabla N\approx\varepsilon e^{-T}$, where $\varepsilon$ denotes the smallness of the initial perturbation, then leads to a small growth of the $L^2$-Sobolev energy of the distribution function as $e^{\varepsilon T}$. The problem is then apparent if this growth of the matter perturbation couples back into the lapse equation, where it reduces the decay of the gradient of the lapse to $\varepsilon e^{(-1+\varepsilon)T}$. This cannot be closed in the sense of a suitable bootstrap argument or by an appropriate energy estimate. A correction mechanism for the $L^2$-Sobolev energy of the distribution function as used to deal with problematic shift vector terms in the 2+1-dimensional case in \cite{Fa16} seems unavailable as the critical terms here do not necessarily appear as an explicit time derivative, which allowed for the correction in \cite{Fa16}.

\subsubsection{A new estimate for the energy density}
We resolve the problem of slow decay of the lapse gradient  by a different idea. A crucial observation therefore is the fact that the matter term in the lapse equation decomposes as
\eq{
N\tau\eta = N\tau(\rho +\tau^2 \underline{\eta})
}
in rescaled variables, where $\rho$ is the rescaled energy density (cf.~\eqref{matter-rescaled}). For Vlasov matter, the remaining term $\tau^2\underline{\eta}$ has stronger decay properties due to the explicit $\tau$ variable, which can be used to compensate a growth of the $L^2$-Sobolev energy of the distribution function. This implies that accepting a small growth for the $L^2$-Sobolev energy still yields a decay of $\tau^2\underline \eta\leq \varepsilon e^{(-3+\varepsilon)T}$, which is sufficiently fast. The problematic term is in fact the rescaled energy density $\rho$. The crucial idea is not to estimate the energy density by the $L^2$-Sobolev energy of the distribution function but to use an explicit evolution equation for $\rho$, which originates from the divergence identity of the energy momentum tensor, $\nabla_\mu T^{\mu\nu}=0$. One obtains the evolution equation for the energy density or \emph{continuity equation}, which in rescaled form (cf.~Appendix B) reads
\eq{\label{ev-eq-intro}
\p T\rho=(3-N) \rho-X^a\na_a\rho+\tau N^{-1}\na_a(N^2\jmath^a)-\tau^2\frac N3g_{ab}T^{ab}-\tau^2N\Si_{ab}T^{ab},
}
where in the setting we consider the last three terms have improved decay from the additional $\tau$ factors. This seems to be a particular feature of massive collisionless matter but this structure may also be relevant for other massive matter models. If those terms are estimated by the $L^2$-Sobolev energies this additional decay can be used to compensate for the small growth and yields a uniform estimate for the standard Sobolev norm of $\rho$ without the problematic loss. This mechanism allows to close the estimates. It is important to remark that the regularity loss of the evolution equation \eqref{ev-eq-intro} for $\rho$ is compensated by the elliptic regularity of the lapse equation which requires the energy density only at one order of regularity below the top order.

\subsubsection{Structure of the proof}
The small growth of the $L^2$-Sobolev energy of the distribution function, which results from the lapse term, implies that we do not correct this energy as done in \cite{Fa16}, where we required uniform boundedness. The corresponding energy estimates here are done with respect to the rescaled variables and require higher orders of regularity but are except for these aspects similar to the ones in \cite{Fa16}. Also similarly to \cite{Fa16} we consider initial data with compact momentum support. We expect that considering non-compact momentum support results in similar decay properties of the system. However, to analyze this issue in detail another additional structural estimate for the transport equation is necessary which is subject to future works on the topic. Regarding the estimates for the perturbation of the geometry we use energy estimates and elliptic estimates according to the vacuum case \cite{AnMo11}, where in the present case additional terms due to the matter quanitities appear. 
For the sake of brevity we derive most estimates under smallness assumptions on the perturbation, which allows us to suppress higher order terms in the perturbation in the estimates and absorb them into uniform constants.
Global existence is eventually shown by a bootstrap argument, which implies that for a sufficiently small initial perturbation the smallness assumptions persist throughout the evolution and almost optimal decay holds, if we compare with the vacuum case.

\subsubsection*{Main theorem}
We formulate the main theorem using the terminology of the remainder of the manuscript. The theorem is formulated with respect to the rescaled metric and second fundamental form. After the theorem we clarify the notation used therein.

\begin{thm}\label{thm-1}
Let $(M,\ga)$ be a 3-dimensional, compact, Einstein manifold without boundary with Einstein constant $\mu=-\frac29$ and $\boldsymbol\varepsilon_{\mathrm{decay}}>0$.
Then there exists an $\varepsilon>0$ such that the future development of the rescaled initial data $(g_0,k_0,f_0)\in H^{6}(M)\times H^5(M)\times H_{\mathrm{Vl},5,3,\mathrm{c}}(TM)$ at $t=t_0$ with
\eq{
 (g_0,k_0,f_0)\in \mcr B^{\,6,5,5}_{\varepsilon}\left(\gamma, \frac{1}3 \gamma,0\right)
}
under the Einstein-Vlasov system is future complete and the rescaled metric and tracefree fundamental form $(g,\Si)$ converge as
\eq{
(g,\Si){\longrightarrow} \left( \gamma,0\right)\quad \mbox{ for }\quad{\tau\nearrow\,0}
} 
with decay rates determined by $\boldsymbol\varepsilon_{\mathrm{decay}}$ as in \eqref{final decay rates}. In particular, any 3+1-dimensional Milne model is future asymptotically stable for the Einstein-Vlasov system in the class of initial data given above.
\end{thm}
The symbols $(g,k,\Si,f)$ denote the Riemannian metric, the second fundamental form, the tracefree part of $k$ and the distribution function, respectively. $\tau<0$ is the mean curvature and is related to the time variable in \eqref{background} via $t=-3\tau^{-1}$ with $\tau\nearrow 0$ being the future direction. $\mcr B_\varepsilon^{6,5,5}(.\,,.\,,.)$ denotes the ball of radius $\varepsilon$ centered at the argument in the set of $H^6(M)\times H^5(M)\times H_{\mathrm{Vl},5,3}(TM)$ with the canonical Sobolev norms defined further below. Here, $H_{\mathrm{Vl},5,3}$ denotes the space of distribution functions on $TM$ corresponding to the standard $L^2$-Sobolev norms, cf.~\cite{Fa15}. $H_{\mathrm{Vl},5,3,\mathrm{c}}(TM)$ is the subset of this space with distribution functions of compact momentum support.

\subsection{Remarks}
The decay rates \eqref{final decay rates} can be achieved for arbitrarily small $\boldsymbol{\varepsilon}_{\mathrm{decay}}$ by choosing the perturbation sufficiently small depending on $\boldsymbol{\varepsilon}_{\mathrm{decay}}$. This implies that one can get arbitrarily close to the vacuum decay rates which correspond to the case $\boldsymbol{\varepsilon}_{\mathrm{decay}}=0$.\\
The corresponding higher dimensional stability results, which for the vacuum equations have been considered in \cite{AnMo11}, are likely resolvable similarly to the case presented herein. In particular, the decay of the matter quantities is expected to be stronger than in the present case. In this sense the 3+1-dimensional case is more difficult.
\subsection{Overview on the paper}
The remainder of the paper is concerned with the proof of Theorem 1. To simplify the presentation we derive all estimates - hyperbolic and elliptic ones - under smallness assumptions on the solution. These smallness assumptions are compatible with the decay properties of the system and this consistency is then shown in the course of a bootstrap argument. In section 2 we discuss the eigenvalue estimate for the Einstein operator for 3-dimensional negative Einstein metrics, recall the rescaling for the Einstein equations and introduce the rescaling for the matter variables. All relevant equations are collected in section 2 and referred to in the course of the following sections. In section 3 we introduce all relevant norms for the geometric quantities and for the distribution function. In view of these, we introduce the notion of smallness which is a prerequisite for establishing all estimates to follow in their respective concise versions. In the global existence argument this notion of smallness is realized in terms of a suitable bootstrap assumption (cf.~\eqref{wefoh}). In sections 4 and 5 we prove the $L^2$-energy estimate and the evolutionary inequality for the bound on the momentum support, respectively. In section 6 we derive the direct energy estimate for the standard Sobolev norm of the energy density $\rho$ of the distribution function. In section 7 we prove elliptic estimates for lapse and shift and their time derivatives. Section 8 contains the energy estimate for the perturbation of the metric and the tracefree part of the second fundamental form. In section 9 we  use the elliptic estimates to reduce all evolutionary estimates to a system of estimates solely containing metric, second fundamental form and matter quantities. Basing on these estimates section 10 presents the proof on Theorem 1, which also contains a number of technical remarks on local existence and existence of initial data in the appropriate sense. The appendix contains a collection of formulae used throughout the paper.
\subsection*{Acknowledgements}
This project results from early discussions of the authors during the conference \emph{Complex Analysis and Dynamical Systems} in Acre 2011. D.~F.~is grateful to Klaus Kr\"oncke for discussions on his results on eigenvalues of the Einstein operator in \cite{Kr15}. D.~F.~ furthermore gratefully acknowledges the support of the Austrian Science Fund (FWF) through the START-Project Y963-N35 of Michael Eichmair as well as through the Project \emph{Geometric transport equations and the non-vacuum Einstein flow} (P 29900-N27). D.F.~acknowledges the hospitality of the \emph{Erwin-Schr\"odinger Institute Vienna} during the program \emph{Geometric Transport equations in General Relativity}.

\section{Preliminaries}
We fix for the remainder of the paper a 3-dimensional Einstein manifold $(M,\gamma)$ with
\eq{
Ric[\ga]=-\frac{2}{9}\ga.
}
\subsection{3-dimensional negative Einstein metrics}
Necessarily, $\ga$ is of constant scalar curvature
\eq{
R[\ga]=-\frac{2}{3}.
} 
We consider the Einstein operator associated with $\ga$,
\eq{
\Delta_E\equiv\nabla^*\nabla -2\overset{\circ}{R},
}
where $\overset\circ{R} h_{ij}=R_{ikjl}h^{kl}$ for symmetric 2-tensors $h$ (cf.~Chapter12D of \cite{Be08} for more details). The lowest positive eigenvalue of $\Delta_E$ plays a crucial role for the construction of suitably decaying energies in the stability problem for the vacuum Einstein flow as demonstrated in \cite{AnMo11}. A similar consideration will be relevant for the nonvacuum problem considered below. We denote the lowest positive eigenvalue of $\Delta_E$ by $\la_0$. The following is an immediate consequence of Kr\"oncke's lower bound on eigenvalues of the Einstein operator (cf.~\cite{Kr15}).

\begin{prop}
Let $(M,\ga)$ be a hyperbolic Einstein 3-manifold with Einstein constant $\mu=-2/9$. Then
\eq{\label{ev-est}
\la_0\geq\frac{1}{9}.
}
\end{prop}

\begin{proof}
From Proposition 3.2 \cite{Kr15} we deduce that the smallest eigenvalue of $\Delta_E\big|_{TT}$, $\Delta_E$ restricted to TT-tensors on $(M,\gamma)$, which we denote  by $\la_{0,TT}$, obeys
\eq{
\la_{0,TT}\geq \frac19.
}
This holds, as $\ga$ is necessarily of constant scalar curvature and therefore has vanishing Weyl tensor. \\
We show that this can be upgraded to \eqref{ev-est} as follows. We observe that if an eigenvalue $\la$ of $\De_{E}$ obeys $(2\mu+\la)<0$ or with the present choice $\la<\frac49$, then its corresponding eigentensor $h_\la$ is TT. This follows as in the proof of Lemma 2.7 in \cite{AnMo11}. In particular, the lowest eigenvalue $\la_0$ either fulfills $\la_0\geq 4/9$ or is in the spectrum of $\Delta_E\big|_{TT}$ and in turn fulfills $\la_0\geq1/9$. 
\end{proof}

A relevant corollary of the above reads

\begin{cor}\label{cor-kernel}
Let $(M,g)$ be a 3-dimensional Einstein manifold with Einstein constant $\mu=-2/9$, then
\eq{
\ker \Delta_E=\{0\}.
}
\end{cor}
This condition assures that the energy to control the perturbation of the geometry defined below is coercive and allows to avoid introducing a shadow-gauge analog to \cite{AnMo11}.  


\subsection{Variables and setup}
We use standard index conventions. Roman letters denote spatial indices $\{1,2,3\}$ and greek letters denote spacetime indices $\{0,1,2,3\}$. In addition, we use bold roman letters to denote indices on the tangent bundle of $TM$. This notation is introduced in Section \ref{sec : L2distr}.
\subsubsection{Standard variables and gauge}
We consider the 3+1-dimensional spacetime in the standard form $(\overline M,\overline g)=(\mathbb R\times M,-\widetilde N^2dt\otimes dt+\widetilde g_{ab}(dx^a+\widetilde X^adt)\otimes(dx^b+\widetilde X^bdt))$, where $\widetilde N$, $\widetilde g$ and $\widetilde X$ denote the lapse function, the induced Riemannian metric on $M$ and the shift vector field\footnote{Note that the coordinate $t$ here does not coincide with the same symbol in the explicit Milne model \eqref{background}}. For the derivation of the Einstein equations in ADM formalism we refer to \cite{Re08}. We denote by $\tau$ the trace of the second fundamental form $\widetilde k$ with respect to $g$ and decompose $\tilde k=\tilde\Si+\frac\tau 3\tilde g$. We then impose the CMCSH gauge via
\eq{\alg{\label{gauges}
t=\tau\\
\widetilde g^{ij}(\widetilde\Gamma^{a}_{ij}-\widehat\Gamma^a_{ij})&=0,
}}
where $\widetilde\Gamma$ and $\widehat\Gamma$ denote the Christoffel symbols of $\widetilde g$ and $\gamma$, respectively.
\subsubsection{Rescaled variables and Einstein's equations}

We rescale the geometry with respect to the mean curvature function $\tau$ analogous to the vacuum case \cite{AnMo11}. This leaves explicit time-factors as coefficients of the matter variables. We rescale those by rescaling the $\tilde p$-variables (cf.~section \ref{sec-Vlasov} ). The variables with respect to mean curvature time $t=\tau$ are denoted by $(\widetilde g,\widetilde \Si, \widetilde N, \widetilde X)$, while the rescaled variables are $(g,\Si,N,X)$. We rescale according to
\eq{\label{rescaling}
\begin{array}{cc}
g_{ij}=\tau^2\tilde{g}_{ij}&N=\tau^2\tilde{N}\\
g^{ij}=\tau^{-2}\tilde{g}^{ij}& \Si_{ij}=\tau\tilde \Si_{ij}\\
p^a=\tau^{-2}\tilde{p}^a& X^i=\tau \tilde{X}^i
\end{array}
}
so the spacetime metric takes the form
\eq{
\overline g=-\tau^{-4}N^2d\tau^2+\tau^{-2}g_{ij}(dx^i+\tau^{-1}X^id\tau)\otimes(dx^j+\tau^{-1}X^jd\tau).
}
Then we  introduce the logarithmic time $T=-\ln(\tau/\tau_0)$, ($\leftrightarrow\tau=\tau_0\exp(-T)$) with
\eq{\label{log-time}
\p T=-\tau \p \tau,
}
which implies
\eq{
\overline g=-\tau^{-2}\tau_0^{-2}N^2dT^2+\tau^{-2}g_{ij}(dx^i+\tau_0^{-1}X^idT)\otimes(dx^j+\tau_0^{-1}X^jdT).
}
We use the notation $\dot{X}=\p TX$, $\dot N=\p TN$ for convenience throughout the manuscript. Also, we denote $\widehat N=\frac N3-1$ and $\widehat X=X/N$. After these modifications the Einstein equations in CMCSH gauge with respect to the rescaled variables take the following form.
\begin{eqnarray}
R(g)-\abg{\Si}^2+\tfrac {2}{3}&=&4\tau\cdot\rho \label{Co-Ha}\\
\na^a\Si_{ab}&=&\tau^2 \jmath_b\label{Co-Mo}\\
\left(\De-\tfrac13\right) N&=& N\Big(\abg{\Si}^2+\underbrace{\tau\cdot \eta}_{(\star)}\Big)\label{lapse}\\
\De X^i+R^i_{\, m}X^m&=&2\na_jN\Si^{ji}-\underbrace{\na^i\left(\tfrac N3-1\right)}_{(\star)}+2N \tau^2{\jmath^i}\label{shift}\\
&&-(2N\Si^{mn}-\na^mX^n)(\Chr imn-\hChr imn)\nonumber\\
\p T g_{ab}&=&\underbrace{2N\Si_{ab}}_{(\star\star)}+2\left(\tfrac{N}3-1\right)g_{ab}-\mcr{L}_Xg_{ab}\label{ev-g}\\
\p T\Si_{ab}&=&\underbrace{-2\Si_{ab}-N\left(R_{ab}+\tfrac{2}{9}g_{ab}\right)}_{(\star\star)}\label{ev-k}\\
&&+\na_a\na_bN+2N\Si_{ai}\Si^{i}_b\nonumber\\
&&-\tfrac13\left(\tfrac{N}{3}-1\right)g_{ab}-\left(\tfrac{N}{3}-1\right)\Si_{ab}\nonumber\\
&&-\mcr{L}_X\Si_{ab}+ \underbrace{N\tau\cdot S_{ab}}_{(\star)}\nonumber
\end{eqnarray}
We denote by $\mathscr L_X$ the Lie-derivative with respect to $X$. Moreover, we recall the decomposition of the curvature term in the spatial harmonic gauge (cf.~\cite{AnMo11}),
\eq{
R_{ab}+\frac29g_{ab}=\frac12\mathcal L_{g,\gamma}(g-\gamma)_{ab}+J_{ab},
}
where 
\eq{
\Ab{J}_{H^{s-1}}\leq C\Ab{g-\gamma}_{H^s}.
}

The rescaled matter quantities are connected to the unrescaled versions via
\eq{\label{matter-rescaled}
\alg{
\rho&:=4\pi\tilde{\rho}\cdot \tau^{-3}\\
\eta&:=4\pi(\tilde{\rho}+\tilde{g}^{ab}\tilde{T}_{ab})\cdot \tau^{-3}\\
\underline{\eta}&:=4\pi \tilde g^{ab}\tilde T_{ab}\cdot\tau^{-5}\\
j^b&:=8\pi \ab{\tau}^{-5} \tilde j^b\\
{S}_{ab}&:=8\pi{\tau}^{-1}\left[\tilde{T}_{ab}-\tfrac1{2}\tilde{g}_{ab}\tilde{T}\right].
}
}
We recall that $\tilde \rho=\tilde N^2\widetilde T^{00}$ is the energy density and $\tilde \jmath^a=-\tilde N\widetilde T_0^a$ is the matter current.
 We also denote $\underline T^{ab}=T^{ab}=\ab{\tau}^{-7}\tilde T^{ab}$ for later purposes. An important identity, which follows immediately from the definitions above is
\eq{\label{dec-eta}
\eta=\rho+\tau^2\underline\eta.
}  
The decomposition is crucial since the second term on the right-hand side in \eqref{dec-eta} decays fast while the first term is handled differently using the continuity equation as explained in the introduction.
\begin{rem}
The right-hand sides of the elliptic system for lapse and shift as well as those for the evolution equations decouple into principal terms and terms which can be considered as perturbative and which turn out to decay faster than the leading order terms.  To give some orientation about which terms are considered to be principal terms, we have marked those terms by the symbol $(\star)$ or $(\star\star)$. The latter case refers to those terms, which are relevant to establish the decay for the energy measuring the perturbation of the geometry. Terms denoted by $(\star)$ are for different reasons principal. In the lapse equation, the $\rho$-term within the $\eta$-term has the slowest decay, while in the shift equation, precisely this slow decay is inherited from the lapse equation through the $(\star)$ term therein. The final principal term to consider is the one in the equation for $\Si$, where due to regularity conditions we cannot estimate the $\rho$-term in $S$ by the $\rho$-energy but we have to use the $L^2$-Sobolev energy of the distribution function to estimate this term. This results in a small loss of decay, which is the reason why this term is of worst decay in the respective equation. 
\end{rem}

%

\subsection{Vlasov matter}\label{sec-Vlasov}

We introduce the structures relevant to Vlasov matter and then rescale the energy-momentum tensor and transport equation according to the previous section.

\subsubsection{The mass-shell relation}

We consider particles of positive mass $m=1$ modeled by distribution functions with domain being the mass-shell
\eq{\label{mass-shell}
\mcr P=\Big\{(x,\mbf p)\in T\overline M\big|\ab{\mbf p}^2_{\overline g}=-1, p^0<0 \Big\},
}
where $\mbf p=\tilde p^\al\partial_\al$. In particular, $\tilde p^\alpha$ are canonical coordinates on the tangent bundle.  We use the $\widetilde{.}$--notation, since below we introduce rescaled variables.  A distribution function $\overline f:\mcr P\rightarrow [0,\infty)$ has the associated energy-momentum tensor
\eq{
\widetilde T^{\al\be}[\overline f] (x)=\int_{\mcr P_x} \overline f\tilde p^\al \tilde p^\be \mu_{\mcr P_x},
} 
where $\mu_{\mcr P_x}$ is the volume form corresponding to the induced metric on $\mcr P_x$. We consider the projection of the distribution function under $\pi : (t,x,p^0,p)\rightarrow (t,x,p)$, which is $f=\overline f\circ (\pi\big |_{\mcr P})^{-1}$, and which we refer to as \emph{distribution function} in the following. We rescale the momentum variables according to
\eq{\label{resc-mom}\alg{
\tilde p^0=p^0,\quad \tau^2p^a=\tilde{p}^a,\quad \p {\tilde p^a}=\tau^{-2} \p {p^a}.
}} 

Then we express the unrescaled mass-shell relation in \eqref{mass-shell} in coordinates (cf.~for instance Section iv in \cite{SaZa14}, equation (37)), which reads
\eq{
\tilde p^0=(\tilde N^{2}-\ab{\widetilde{X}}_{\tilde g})^{-1}\left(\widetilde{ X}_j\tilde p^j+\sqrt{(\widetilde{ X}_j\tilde p^j)^2+(\tilde N^{2}-\ab{\widetilde{ X}}_{\tilde g})(1+\ab{\tilde p}_{\tilde g}^2)}\right),
}
and replace all variables by their rescaled counterparts. This yields an expression for $p^0$ as a function of the $p^a$ variables and the metric components in the form
\eq{\label{msr-resc}
p^0=N^{-1}(1-\abg{\hat{ X}}^2)^{-1}\Big[\tau\hat X_jp^j+\sqrt{\tau^2(\hat X_jp^j)^2+(1-\abg{\hat{ X}}^2)(1+\tau^2\abg{p}^2)}\Big].
} 

An alternative expression is given by
\eq{
p^0=\frac1N\frac{1}{\hat p -\tau\langle \hat X , p \rangle_g}(1+\tau^2\abg{p}^2),
}
where 
\eq{
\widehat p=\sqrt{\tau^2(\hat X_jp^j)^2+(1-\abg{\hat{ X}}^2)(1+\tau^2\abg{p}^2)}
}
is just defined for convenience and does not necessarily have a specific geometric meaning. In addition, $p_0=\overline g_{0\nu}\tilde p^\nu=-N\widehat p$.
We derive some useful estimates for $p^0$ using elementary manipulations.
We furthermore use the simplifying notation
\eq{
\underline p=Np^0.
}
\begin{rem}
Note, that the rescaled mass-shell relation \eqref{msr-resc} reduces to $p^0=\sqrt{1+\tau^2\abg{p}^2}$, when $X=0$, $N=1$, which corresponds to the background solution. In particular, the constant term under the squareroot, which originates from the mass term, scales like a constant, while the second term decays fast in expanding direction ($\tau\nearrow0$).
\end{rem}
The following lemma contains two useful pointwise estimates on the momentum variable $p^0$.
\begin{lem}
\begin{align}\label{p-est1}
\frac{\abg p}{p^0} &\leq 2\abg X\ab{\tau}^{-1}+N\sqrt{1-\abg {\hat X}^2}\ab{\tau}^{-1}\\
p^0&\leq \frac{1}{N}\frac{1}{1-\abg{\hat X}^2}\left[2\ab\tau\abg{\hat X}\abg p+\sqrt{1-\abg{\hat X}^2}\sqrt{1+\tau^2\abg p^2}\right]\label{p-est2}
\end{align}
\end{lem}

\subsubsection{The transport equation}
We introduce the transport equation and its rescaling. The transport equation
\eq{
\tilde{p}^\al \p \al \tilde{f}-\Chrt a\mu\nu\tilde{p}^\mu\tilde{p}^\nu\p {\tilde{p}^a}\tilde{f} =0
}
is rescaled via \eqref{resc-mom}. To express the transport equation only in terms of the rescaled variables, we require the rescaled Christoffel symbols $\Gamma$.
The non-rescaled Christoffel symbols read (cf.~\cite{Re08})
\begin{align} 
 \Chrtn abc4&=\Chrtg abc+\tilde N^{-1}\tilde k_{bc}\tilde X^a,\\
 \Chrtn a004&=\p t\tilde X^a+\tilde X^b\na_b\tilde X^a-2\tilde N\tilde k_c^a\tilde X^c+\tilde{N}\tilde{\na}^a\tilde{N}\\
 &\nonumber\quad-\tilde N^{-1}(\p t\tilde N+\tilde X^b\p b\tilde N-\tilde k_{bc}\tilde X^b\tilde X^c)\tilde X^a,\\
 \Chrtn a0b4&=-\tilde N\tilde k^a_b+\nabla_b\tilde X^a-\tilde N^{-1}\tilde X^a\nabla_b\tilde N+\tilde N^{-1}\tilde k_{bc}\tilde X^c\tilde X^a.
 \end{align}
In terms of the rescaled variables the Christoffel symbols are of the form

\eq{\alg{
 \Chrtn abc4&=\Chr abc(g)+N^{-1}\left(\Si_{bc}+\frac{1}{3}g_{bc}\right)X^a,\\
 \Chrtn a004 &=-\p TX^a+\tau^{-2}\Gamma^a,\\
 \Chrtn a0c4&=\tau^{-1}\left(-\delta_c^a+\Gamma_c^a\right),
}}

where we denote

\eq{\alg{
\Gamma^a&=-X^a-\tfrac23(N-3)X^a+X^b\na_b X^a-2 N \Si_c^a X^c+{N}{\na}^a{N}\\
 &\quad\,+ \left(N^{-1}\p T N-N^{-1} X^b\p b N+N^{-1}\left(\Si_{bc}+\tfrac13g_{bc}\right) X^b X^c\right)X^a,\\
\Gamma^a_c&=-N\Si^a_b+\delta_c^a(1-\tfrac{N}{3})+\nabla_b X^a- N^{-1} X^a\nabla_b N+ N^{-1}\left( \Si_{bc}+\tfrac13g_{bc}\right) X^c X^a.
}}
We refer to the latter terms also by the symbols $\Gamma^*$ and $\Gamma^*_*$, respectively, when the indices are suppressed.
The fully rescaled transport equation then reads

\eq{\label{eq-transport}\alg{
\p Tf &= \tau N p^a/\underline{p}\mbf A_ af-\frac{\underline{p}}{N}\Big[-\tau\p TX^a+\underbrace{\tau^{-1}\Gamma^a}_{(\star)}\Big]\B_af +\underbrace{2p^i\B_if}_{(\star\star)}\\
&\quad-2\Gamma_u^ep^u\B_ef- \left(\Si_{ab}+\frac{1}{3}g_{ab}\right)X^e\frac{p^ap^b}{\underline{p}}\B_ef,
}} 
where we denote
\eq{\alg{
\A_a&=\p a-p^i\Chr kai\B_k,\\
\B_a&=\pp a,
}}
which correspond to the natural horizontal and vertical derivatives on $TM$. The two marked terms are leading order in the sense that term $(\star)$, among the small terms, has the slowest decay as $\Gamma^a$ contains in particular $\nabla N$, which in combination with $\tau^{-1}$ is of the order of $\varepsilon$. Term $(\star\star)$ is the dilution term, driving the downscaling of the momentum support in expanding direction of spacetime and thereby the dilution of the matter variables.

\subsection{Energy momentum tensor}

The rescaled matter quantities as appearing in the Einstein equations take the following form in terms of the distribution function $f$.
\begin{eqnarray}
\rho(f)&=& N^2\Ip{f \frac{(p^0)^2}{\underline{p}} \sg}\\
\jmath^a(f)&=&  N\Ip{f\frac{p^a p^0}{\underline{p}}\sg} \\
\underline{\eta}(f)&=&  \Ip{f\frac{\abg{p+\tau^{-1}p^0X}^2}{\underline{p}}\sg}\\
\underline T^{ab}(f)&=&\int f \frac{p^ap^b}{\underline p}\sg dp\\
S_{ab}(f)&=&\ab\tau^2\int f\frac{(p^ig_{ia}+\ab\tau^{-1}p^0X_a)(p^jg_{ja}+\ab\tau^{-1}p^0X_b)}{\underline p}\sg dp\\
&&\nonumber+\frac12g_{ab}\rho(f)-\frac12 g_{ab}\ab\tau^2 \underline\eta(f).
\end{eqnarray}
\begin{rem}
Recall the definition of the rescaled energy-momentum tensor \eqref{matter-rescaled} for the identities above. Note, that the expressions for the energy-momentum variables are a consequence of their definition \eqref{matter-rescaled} and the rescalings \eqref{rescaling} and \eqref{resc-mom}.
\end{rem}

\subsection{Preview on the decay rates}
For the study of the estimates to follow it is important to have an idea about smallness and decay of the rescaled quantities. The quantities
\eq{
N-3,\quad \Si,\quad g-\gamma\,\, \mbox{  and  }\,\, X
}
are small and decay with the rates
\eq{\alg{
\Abk{\Si}5^2+\Abk{g-\gamma}6^2&\lesssim\varepsilon \exp\left(-(2-\boldsymbol{\varepsilon}_{\mathrm{decay}})T\right)\\
\Abk{N-3}{6}+\Abk{X}6&\lesssim \varepsilon \exp(-T).
}}
For the matter terms we have 
\eq{\alg{
\Abk{\rho}4&\lesssim\varepsilon,\\
\Abk{\rho}5+\Abk{\jmath}5+\Abk{\underline\eta}5+\Abk{\underline T}5+\Abk{S}5&\lesssim \varepsilon \exp\left({C\varepsilon T}\right).
}}

These decay rates are shown to be valid for sufficiently small initial data, where $\varepsilon$ is the smallness of the initial perturbation and $\boldsymbol{\varepsilon}_{\mathrm{decay}}>0$ can be chosen arbitrarily small.


\section{Norms and smallness}

We introduce all relevant norms for measuring the perturbation of the geometry and the distribution function. Some norms are defined with respect to the fixed Einstein metric $\ga$ and others are defined with respect to the rescaled dynamical norm $g$. As we impose a uniform smallness assumption all these norms are equivalent. We assume for the remainder of the paper that $T_0>1$.

\subsection{Constants}
We use the symbol $C$ to denote any positive constant, which is uniform in the sense that it does not depend on the solution of the system once a smallness parameter $\varepsilon$ for the initial data and an initial time $T_0$ are chosen. Furthermore, if $\varepsilon$ is further decreased or $T_0$ is increased, $C$ keeps its value. 

\subsection{Norms - tensor fields}
For functions and symmetric tensor fields on $M$ we denote the standard Sobolev norm with respect to the fixed metric $\ga$ of order $\ell\geq0$ by $\Abk{.}\ell$. The corresponding function spaces are denoted by $H^{\ell}=H^{\ell}(M)$.

\subsection{Norms - distribution function}
We introduce different metrics on $TM$ and related notation necessary for the definition of $L^2$-Sobolev energies for the distribution function. This construction is based on the the metric $\gamma$ on $M$. In the following section we consider the case when the corresponding construction is based on the rescaled metric $g$. \\

\indent The metric $\gamma$ induces the related Sasaki metric $\boldsymbol \gamma$ on $TM$ via
\eq{
\boldsymbol \gamma\equiv \gamma_{ij} dx^i\otimes dx^j+\gamma_{ij}Dp^i\otimes Dp^j,
}
where $Dp^i=dp^i+\widehat{\Gamma}^i_{jk}p^jdx^k$. Recall that $\widehat\Gamma$ denotes the Christoffel symbols of $\gamma$. The covariant derivative on the tangent bundle corresponding to $\boldsymbol \gamma$ is denoted by $^{\boldsymbol \gamma}\boldsymbol{\na}$. We consider the volume form on $TM$,
\eq{
\mu_{\boldsymbol \gamma}=\ab{ \gamma}dx^3 \wedge dp^3.
}
We define a weighted version of the Sasaki metric by 
\eq{
\underline{\boldsymbol \gamma}= \gamma_{ij}dx^i\otimes dx^j+\overline p_{\gamma}^{-2}\gamma_{ij}Dp^i\otimes Dp^j,
}

where we denote $\overline p_\gamma=\sqrt{1+\ab{p}_{\gamma}^2}$. This metric is necessary to take the norm in the energies to be defined below, which require a weight in the momentum-direction. We define the $L^2$-Sobolev energy of the distribution function with respect to Sasaki metric corresponding to the fixed metric $\gamma$ by
\eq{
|\!|\!|f|\!|\!|_{\ell,\mu} \equiv \sqrt{\sum_{k\leq \ell}\int_{TM}\overline p_{\gamma}^{2\mu+4(\ell-k)}
 |^{\boldsymbol\gamma}\boldsymbol{\nabla}^kf|^2_{\underline {\boldsymbol \gamma}}\mu_{\boldsymbol{\gamma}}}.
}
The corresponding function spaces are denoted by $H_{\mathrm{Vl},\ell,\mu}(TM)$. Pointwise estimates are taken with respect to the following $L^\infty_xL_p^2$-norm,
\eq{
|\!|\!|f|\!|\!|_{\infty,\ell,\mu}\equiv \sup_{x\in M}\left\{\sqrt{\int_{T_xM}\overline p_{\gamma}^{2\mu+4(\ell-k)}
 |^{\boldsymbol\gamma}\boldsymbol{\nabla}^kf|^2_{\underline {\boldsymbol \gamma}}\sqrt{\gamma} dp}\right\},
}
which obeys the following lemma.
\begin{lem}
For $f$ sufficiently regular
\eq{
|\!|\!|f|\!|\!|_{\infty,\ell,\mu}\leq C |\!|\!|f|\!|\!|_{\ell+2,\mu}
}
holds.
\end{lem}

\subsection{Smallness}
We define a set of smallness conditions for the dynamical quantities. These are designed to include weights in terms of the time-function to incorporate some decay properties indirectly. These are chosen in a way that in the proof of global existence the smallness conditions serves as a part of the bootstrap assumptions and leaves room to be improved for sufficiently small data and sufficiently large times. We define
\eq{\alg{
&\mcr B^{6,5,5}_{\delta,\tau}\left(\gamma,0,0\right)\\
&\equiv\Big\{(g,\Si,f)\in H^6\times H^5\times H_{\mathrm{Vl},5,4}\Big|\\
&\qquad\qquad \sqrt{\ab\tau}^{-1}(\Abk{g-\gamma}6+\Abk{\Si}5)+\Abk{\rho(f)}4+\sqrt{\ab\tau}|\!|\!|f|\!|\!|_{5,4}<\delta\Big\}.
}}

We say a triple $(g(\tau),\Si(\tau),f(\tau))$ is \emph{$\delta$-small} when $(g,\Si,f)\in\mcr B^{6,5,5}_{\delta,\tau}(\gamma,0,0)$. Also, we mostly suppress the dependence on $(\gamma,0,0)$ in the notation. In addition we use the term \emph{smallness assumptions} if we refer to $\delta$-small data.

\subsection{Some immediate estimates}

Smallness in the above sense implies smallness of the perturbation for lapse function and shift vector. The following corollary uses the elliptic estimates proven in section \ref{sec : ell}.

\begin{cor}
For any $\delta>0$ there exists a $\overline{\delta}$ such that 
\eq{
(g,\Si,f)\in \mcr B^{6,5,5}_{\overline\delta,\tau} \Rightarrow \ab{\tau}\left(\Abk{N-3}6+\Abk{X}6\right)<\delta.
}

\end{cor}

\begin{proof}
This is an immediate consequence of Proposition \ref{prop-ell-est}.
\end{proof}

\section{$L^2$ -- estimates for the distribution function}\label{sec : L2distr}
We define the $L^2$-Sobolev energy for the distribution function in terms of the Sasaki metric associated with $g$. Under the present smallness assumptions this energy is equivalent to the norm $|\!|\!|f|\!|\!|_{\ell,\mu}$. We define the corresponding metrics on $TM$ with respect to $g$ as follows. The metric $g$ induces the related Sasaki metric $\mbf g$ on $TM$ via
\eq{
\mbf g\equiv g_{ij} dx^i\otimes dx^j+g_{ij}Dp^i\otimes Dp^j,
}
where $Dp^i=dp^i+\Gamma^i_{jk}p^jdx^k$. The covariant derivative corresponding to $\mbf g$ is denoted by $\boldsymbol{\na}$. We consider the volume form on $TM$,
\eq{
\mu_{\mbf g}=\ab{g}dx^3 \wedge dp^3.
}
We define a weighted version of the Sasaki metric by 
\eq{
\underline{\mbf g}=g_{ij}dx^i\otimes dx^j+\overline p^{-2}g_{ij}Dp^i\otimes Dp^j,
}

where we denote $\overline p=\sqrt{1+\abg{p}^2}$. For explicit computations including the Sasaki metric on the tangent bundle we use indices $\mbf a$, $\mbf b$, $\hdots \in\{1,\hdots, 6\}$, where $1,2,3$ correspond to horizontal directions and $4,5,6$ to vertical directions. We introduce the frame $\{\theta_{\mbf a}\}_{\mbf a\leq 6}=\{\mbf A_1,\A_2,\A_3,\B_1,\B_2,\B_3\}$ and we denote the connection coefficients of the Sasaki metric in this frame by $\boldsymbol \Gamma$ (cf.~\eqref{co-coeff}). We define the $L^2$-Sobolev energy of the distribution function by
\eq{
\mcr E_{\ell,\mu} (f)\equiv \sqrt{\sum_{k\leq \ell}\int_{TM}\overline p^{2\mu+4(\ell-k)}
 |\boldsymbol{\nabla}^kf|^2_{\underline {\mbf g}}\mu_{\mbf g}}.
}
\begin{rem}
The choice for the weights, to increase with decreasing level of regularity, is necessary to absorb terms with high weights, which result from the connection coefficients of the Sasaki metric, where momentum weights appear in conjunction with horizontal derivatives. This is discussed in more detail below.\end{rem}

\begin{lem}
For $\delta$-small data with $\delta$ sufficiently small the energies $\mcr E_{5,4}(f)$ and $|\!|\!|f|\!|\!|_{5,4}$ are equivalent.
\end{lem}

The main energy estimate for the distribution function is given in the following.

\begin{prop}\label{prp-l2-est}
For $\delta>0$ sufficiently small, $(g,\Si,f)\in\mcr B^{6,5,5}_{\tau,\delta}$ and $f$ a solution to the transport equation \eqref{eq-transport} the following estimate holds.
\eq{\label{en-est-f-prp}\alg{
\p T \mcr E^2_{\ell,\mu}(f)&\leq C\Big[ \Abk{N-3}{\ell}+\Abk{\Si}{\ell}+\Abk{X}{\ell+1}+\ab{\tau}\Abk{N^{-1}\p T X}{\ell}+\ab{\tau^{-1}}\Abk{N^{-1}\Gamma^*}{\ell} \\
&\qquad\,\,+\Abk{\Gamma^*_*}{\ell}+\Abk{(\Si+g)X}{\ell}+\ab\tau\mcr G\Big]\cdot   \mcr E^2_{\ell,\mu}(f)
}}
\end{prop}

\begin{proof}
The derivation of the energy estimate is a straightforward and technical computation. It follows the lines of the analogous computations in (\cite{Fa16}, section 5). We discuss some exemplary steps, which are more general herein.\\

We take the time derivative of the square of the energy, which yields four leading order terms. The first term results from the time-derivatives hitting the distribution function and reads
\eq{\label{d1-1}
2\int_{TM}\overline p^{2\mu+4(\ell-k)} \underline{\mbf g}^{\mbf a_1\mbf b_1 }\cdot\hdots\cdot \underline{\mbf g}^{\mbf a_\ell\mbf b_\ell}\boldsymbol{\nabla}_{\mbf a_1}\hdots\boldsymbol{\nabla}_{\mbf a_\ell}f\cdot \p T \boldsymbol{\nabla}_{\mbf b_1}\hdots\boldsymbol{\nabla}_{\mbf b_\ell}f \mu_{\mbf g}.
}
The second leading order term arises from the time-derivative of the volume form, which, when the derivative acts on the time-function in the rescaled momentum variables and reads
\eq{\label{dfiuh}
6\cdot\mcr{E}_{\mu,\ell}^2(f),
}
since $\boldsymbol\mu_{TM}=\ab{g}\tau^{-6}d^3\tilde p$, where $\tilde p$ is time-independent. The third leading order term occurs when the time derivative hits the time function in the momentum-weight factor and yields
\eq{\label{wefugef}
2(2\mu+4(\ell-k)-2)\int_{TM}\overline p^{2\mu+4(\ell-k)-2}\abg{p}^2
|\boldsymbol{\nabla}^kf|^2_{\underline {\mbf g}}\mu_{\mbf g}.
}
The fourth leading order term results from the time-derivative hitting the momentum variable in the inverse $\underline{\mathbf g}^{\mbf a_i\mbf b_i}$, when $\mbf a_i,\mbf b_i\geq 4$ and reads
\eq{\label{sdfhb}
\int_{TM}\frac{4\abg p^2}{1+\abg{p}^2}\overline p^{2\mu+4(\ell-k)}
|\boldsymbol{\nabla}^kf|^2_{\underline {\mbf g}}\mu_{\mbf g}
}
for each pair with $\mbf a_i,\mbf b_i\geq 4$. The non-explicitly listed terms arise when the time-derivative hits the rescaled metric $g$, which yields terms of the first three types listed in \eqref{en-est-f-prp}. We evaluate term \eqref{d1-1} in the following.

\eq{\alg{
\p T \boldsymbol{\nabla}_ {\mbf a_1}\hdots\boldsymbol{\nabla}_ {\mbf a_k}f&=\p T \left(\theta_{\mbf a_1}\boldsymbol{\nabla}_ {\mbf a_2}\hdots\boldsymbol{\nabla}_ {\mbf a_k}f-\sum_{ 2\leq i\leq k}\hChr {\mbf{e}}{\mbf a_i}{\mbf a_1}\boldsymbol{\nabla}_ {\mbf a_2}\hdots\boldsymbol{\nabla}_ {\mbf e}\hdots\boldsymbol{\nabla}_ {\mbf a_k}f \right)\\
&=[\p T,\theta_{\mbf a_1}]\boldsymbol{\nabla}_ {\mbf a_2}\hdots\boldsymbol{\nabla}_ {\mbf a_k}f
-\sum_{ 2\leq i\leq k}(\p T\hChr {\mbf{e}}{\mbf a_i}{\mbf a_1})\boldsymbol{\nabla}_ {\mbf a_2}\hdots\boldsymbol{\nabla}_ {\mbf e}\hdots\boldsymbol{\nabla}_ {\mbf a_k}f \\
&\quad\,\,+\boldsymbol{\nabla}_ {\mbf a_1}\p T \boldsymbol{\nabla}_ {\mbf a_2}\hdots\boldsymbol{\nabla}_ {\mbf a_k}f\\
&=\underbrace{[\p T,\theta_{\mbf a_1}]\boldsymbol{\nabla}_ {\mbf a_2}\hdots\boldsymbol{\nabla}_ {\mbf a_k}f}_{(1)}
-\underbrace{\sum_{ 2\leq i\leq k}(\p T\hChr {\mbf{e}}{\mbf a_i}{\mbf a_1})\boldsymbol{\nabla}_ {\mbf a_2}\hdots\boldsymbol{\nabla}_ {\mbf e}\hdots\boldsymbol{\nabla}_ {\mbf a_k}f}_{(2)} \\
&\quad\,\,+\sum_{j\leq k-2}\boldsymbol{\nabla}_ {\mbf a_1}\hdots\boldsymbol{\nabla}_ {\mbf a_j}\Big(\underbrace{[\p T,\theta_{\mbf a_{j+1}}]\boldsymbol{\nabla}_ {\mbf a_{j+2}}\hdots\boldsymbol{\nabla}_ {\mbf a_k}f}_{(3)}\\
&\qquad\qquad\qquad\qquad\qquad\quad-\underbrace{\sum_{ j+2\leq i\leq k}(\p T\hChr {\mbf{e}}{\mbf a_i}{\mbf a_{j+1}})\boldsymbol{\nabla}_ {\mbf a_{j+2}}\hdots\boldsymbol{\nabla}_ {\mbf e}\hdots\boldsymbol{\nabla}_ {\mbf a_k}f}_{(4)}\Big)\\
&\quad\,\,+\underbrace{\boldsymbol{\nabla}_ {\mbf a_1}\hdots \boldsymbol{\nabla}_ {\mbf a_{k-1}}[\p T,\theta_{\mbf a_k}]f}_{(5)}+\underbrace{\boldsymbol{\nabla}_ {\mbf a_1}\boldsymbol{\nabla}_ {\mbf a_2}\hdots\boldsymbol{\nabla}_ {\mbf a_k}\p T f}_{(6)}\\
}}
We first analyze the terms containing commutators of $\p T$ with $\theta_{\mbf a_i}$. Since $\theta_{\mbf a_i}$ is not affected from the rescaling when $\mbf a_i\leq 3$ (all time-factors cancel) the commutator can be estimated by terms arising from $\p T\Gamma(g)$, which yields terms of the form $\nabla\Si$, $\nabla^2 X$ and $\nabla(N-3)$. When $\mbf a_i\geq 4$, i.e. $\theta_{\mbf a_i}=\B_{\mbf a_i-3}$ we have $[\p T,\theta_{\mbf a_i}]=-2\theta_{\mbf a_i}$. This implies that in the case of $\mbf a_i\leq 3$, term $(1)$ can be estimated by terms included in the first three terms on the right-hand side of \eqref{en-est-f-prp}. In the complementary case this results in a term of the form
\eq{
-2\theta_{\mbf a_1}\boldsymbol{\nabla}_ {\mbf a_2}\hdots\boldsymbol{\nabla}_ {\mbf a_k}f,
}
which requires to be canceled for the estimate to hold, as we see below. The terms $(3)$ and $(5)$ again give rise to terms of the form of the three first terms on the right-hand side of \eqref{en-est-f-prp} if $\mbf a_{j+1}\leq 3$ and $\mbf a_{k}\leq 3$. In the complementary case, terms of the form

\eq{\label{sdflk}
-2\sum_{j\leq k-2}\boldsymbol{\nabla}_ {\mbf a_1}\hdots\boldsymbol{\nabla}_ {\mbf a_j}\theta_{\mbf a_{j+1}}\boldsymbol{\nabla}_ {\mbf a_{j+2}}\hdots\boldsymbol{\nabla}_ {\mbf a_k}f \mbox{ and } -2\boldsymbol{\nabla}_ {\mbf a_1}\hdots\boldsymbol{\nabla}_ {\mbf a_{k-1}}\theta_ {\mbf a_k}f
}
occur, which are canceled by terms arising below. Regarding terms $(2)$ and $(4)$, from \eqref{co-coeff} we observe that these terms yield time derivatives of $\Gamma$ or of the Riemann tensor $\mbf{Riem}$, which in combination again yields terms of the form of the first three terms on the right-hand side of \eqref{en-est-f-prp} and terms that arise when the time derivative hits the rescaled momentum variable in the respective cases in \eqref{co-coeff}. From these we again obtain leading order terms, which are of the form

\eq{
-2\hChr {\mbf{e}}{\mbf a_i}{\mbf a_1}\boldsymbol{\nabla}_ {\mbf a_2}\hdots\boldsymbol{\nabla}_ {\mbf e}\hdots\boldsymbol{\nabla}_ {\mbf a_k}f
}
when $\mbf e\leq 3$ and $4<\mbf a_i+\mbf a_1\leq 10$
or when $\mbf e\geq 4$ and $\mbf a_i, \mbf a_1\leq 3$; and

\eq{\label{qwejb}
-2\boldsymbol{\nabla}_ {\mbf a_1}\hdots\boldsymbol{\nabla}_ {\mbf a_j}
\hChr {\mbf{e}}{\mbf a_i}{\mbf a_{j+1}}\boldsymbol{\nabla}_ {\mbf a_{j+2}}\hdots\boldsymbol{\nabla}_ {\mbf e}\hdots\boldsymbol{\nabla}_ {\mbf a_k}f,
}
when $\mbf e\leq 3$ and $4<\mbf a_i+\mbf a_{j+1}\leq 10$
or when $\mbf e\geq 4$ and $\mbf a_i, \mbf a_{j+1}\leq 3$.
Both types of terms are cancelled by terms arising below.\\

It remains to consider term $(6)$, where $\p Tf$ is replaced by the transport equation yielding the following term.

\eq{\label{dfiu}\alg{
\boldsymbol{\nabla}_ {\mbf a_1}\boldsymbol{\nabla}_ {\mbf a_2}\hdots\boldsymbol{\nabla}_ {\mbf a_k}\Big(&\tau N p^a/\underline{p}\mbf A_ af-\frac{\underline{p}}{N}\Big[-\tau\p TX^a+\tau^{-1}\Gamma^a\Big]\B_af +\underbrace{2p^i\B_if}_{(\star)}\\
&\quad-2\Gamma_u^ep^u\B_ef- \left(\Si_{ab}+\frac{1}{3}g_{ab}\right)X^e\frac{p^ap^b}{\underline{p}}\B_ef\Big)
}}

We begin with the most important term to evaluate, which is here marked by $(\star)$. This term is relevant for the cancellation of all non-perturbative terms above. Before we start the computation we derive a few simple commutators. The following identities hold.
\eq{\alg{
[\B_j,p^i\B_i ]f&= \B_j f\\
[\A_j,p^i\B_i ]f&=0\\
[\boldsymbol{\Gamma}^{\mbf c}_{\mbf a\mbf b},p^i\B_i]f&=\begin{cases}
-\boldsymbol{\Gamma}^{\mbf c}_{\mbf a\mbf b}&\mbox{if }\mbf c\geq 4; \mbf a,\mbf b\leq 3 \mbox{ or }\mbf c\leq 3; 4\leq\mbf a+\mbf b\leq9\\
0& \mbox{ else }
\end{cases}
}}
We evaluate now the term from above.

\eq{\alg{
2\boldsymbol{\nabla}_ {\mbf a_1}\boldsymbol{\nabla}_ {\mbf a_2}\hdots\boldsymbol{\nabla}_ {\mbf a_k}p^i\B_i f=2\boldsymbol{\nabla}_ {\mbf a_1}\boldsymbol{\nabla}_ {\mbf a_2}\hdots \boldsymbol{\nabla}_{\mbf a_{k-1}}p^i\B_i\boldsymbol{\nabla}_ {\mbf a_k} f-2\boldsymbol{\nabla}_ {\mbf a_1}\boldsymbol{\nabla}_ {\mbf a_2}\hdots \boldsymbol{\nabla}_{\mbf a_{k-1}}[p^i\B_i,\boldsymbol{\nabla}_ {\mbf a_k}] f
}}
According to the commutators above, the second term in the previous line vanishes if $\mbf a_k\leq 3$ or cancels the second term in \eqref{sdflk}. We proceed with the first term.
\eq{\alg{
2\boldsymbol{\nabla}_ {\mbf a_1}\boldsymbol{\nabla}_ {\mbf a_2}\hdots \boldsymbol{\nabla}_{\mbf a_{k-1}}p^i\B_i\boldsymbol{\nabla}_ {\mbf a_k} f&=
2\boldsymbol{\nabla}_ {\mbf a_1}\boldsymbol{\nabla}_ {\mbf a_2}\hdots \Big(\theta_{\mbf a_{k-1}}p^i\B_i\boldsymbol{\nabla}_ {\mbf a_k} f-\boldsymbol{\Gamma}_{\mbf a_k\mbf a_{k-1}}^{\mbf e}p^i\B_i\boldsymbol{\nabla}_ {\mbf e} f\Big)\\
&=2\boldsymbol{\nabla}_ {\mbf a_1}\boldsymbol{\nabla}_ {\mbf a_2}\hdots \Big(p^i\B_i\boldsymbol{\nabla}_{\mbf a_{k-1}}\boldsymbol{\nabla}_{\mbf a_{k}}f-[p^i\B_i,\theta_{\mbf a_{k-1}}]\boldsymbol{\nabla}_{\mbf a_{k}}f\\
&\qquad\qquad\qquad\qquad+[p^i\B_i,\boldsymbol{\Gamma}_{\mbf a_k\mbf a_{k-1}}^{\mbf e}]\boldsymbol{\nabla}_{\mbf e}f\Big)
}}
According to the previous step, the second term cancels the corresponding term from \eqref{sdflk} and the third term on the right-hand side cancels the corresponding term from \eqref{qwejb}.\\
Continuing with the first term on the right-hand side of the previous equation and further commuting $p^i\B_i$ to the left, we obtain terms cancelling all terms in \eqref{sdflk} and \eqref{qwejb}. Then we are left with the term
\eq{\alg{
&2\int_{TM}\overline p^{2\mu+4(\ell-k)} \underline{\mbf g}^{\mbf a_1\mbf b_1 }\cdot\hdots\cdot \underline{\mbf g}^{\mbf a_\ell\mbf b_\ell}\boldsymbol{\nabla}_{\mbf b_1}\hdots\boldsymbol{\nabla}_{\mbf b_k}f\cdot
2p^i\B_i \boldsymbol{\nabla}_ {\mbf a_1}\boldsymbol{\nabla}_ {\mbf a_2}\hdots \boldsymbol{\nabla}_{\mbf a_{k-1}}\boldsymbol{\nabla}_ {\mbf a_k}f\mu_{\mbf g}\\
&=2\int_{TM}\overline p^{2\mu+4(\ell-k)} \underline{\mbf g}^{\mbf a_1\mbf b_1 }\cdot\hdots\cdot \underline{\mbf g}^{\mbf a_\ell\mbf b_\ell}p^i\B_i\Big(  \boldsymbol{\nabla}_{\mbf b_1}\hdots\boldsymbol{\nabla}_{\mbf b_k}f\cdot
\boldsymbol{\nabla}_ {\mbf a_1}\hdots \boldsymbol{\nabla}_ {\mbf a_k}f\Big)\mu_{\mbf g}.
}}
Integration by parts yields three types of terms. The first term arises when $\B_i$ acts on $p^i$ and cancels \eqref{dfiuh}. The second term results from $\B_i$ acting on $\overline p$ and cancels \eqref{wefugef}. Finally, the term arising from $\B_i$ acting on $\underline{\mathbf g}^{\mbf a_i\mbf b_i}$ when $\mbf a_i,\mbf b_i\geq 4$ cancels \eqref{sdfhb}.\\

It remains to consider the remaining terms in \eqref{dfiu}. When estimating the term corresponding to the first term in \eqref{dfiu} we use the estimate
\eq{
\frac{\abg{p}}{\underline p}\leq \mcr G.
}
The corresponding term in the estimate \eqref{en-est-f-prp} is $\ab{\tau}\mcr G$.
Note that compact support is necessary for this. Otherwise we would obtain an additional factor $\ab{\tau}^{-1}$, which would leave no decay for this term. To outline the estimates in more detail we consider one particular term from \eqref{dfiu} and claim the other terms can be handled in a similar way.

We sketch
\eq{\alg{
&\boldsymbol{\nabla}_ {\mbf a_1}\boldsymbol{\nabla}_ {\mbf a_2}\hdots\boldsymbol{\nabla}_ {\mbf a_k}\left(p^u\Gamma^e_u\B_e f\right)\\
&=\theta_ {\mbf a_1}\boldsymbol{\nabla}_ {\mbf a_2}\hdots\boldsymbol{\nabla}_ {\mbf a_k}-\sum_{2\leq j \leq k}\boldsymbol{\Gamma}^{\mbf e}_{\mbf a_j\mbf a_1}\boldsymbol{\nabla}_ {\mbf a_2}\hdots\boldsymbol{\nabla}_{\mbf e}\hdots\boldsymbol{\nabla}_ {\mbf a_k}\\
&=\theta_{\mbf a_1}\hdots \theta_{\mbf a_k} (p^u\Gamma^e_u\B_e f)+\hdots+(-1)^{k-1}\sum \boldsymbol{\Gamma}^{\mbf e_1}_{\mbf a_k\mbf a_1}\boldsymbol{\Gamma}_{\mbf e_1\mbf a_2}^{\mbf e_2}\hdots\boldsymbol{\Gamma}^{\mbf e_{k-1}}_{\mbf e_{k-2}\mbf a_{k-1}}\theta_{\mbf e_{k-1}} (p^u\Gamma^e_u\B_e f),
}}
where we suppress all mixed terms. Commuting the operator $p^u\Gamma_u^e\B_e $ to the front we obtain a term of the form
\eq{
p^u\Gamma^e_u\B_e\left(\boldsymbol{\nabla}_ {\mbf a_1}\boldsymbol{\nabla}_ {\mbf a_2}\hdots\boldsymbol{\nabla}_ {\mbf a_k} f\right).
}
The corresponding integral, after an integration by parts, yields the term $\Abk{\Gamma_*^*}\ell$ in \eqref{en-est-f-prp}. The remaining terms, after commuting $p^u\Gamma_u^e\B_e $ to the front, are schematically of the form
\eq{
p^u(\nabla^{k_1} \Gamma^e_u)\left(p^i \nabla^{k_2}\mathrm{Riem}\right)^{k_3} \boldsymbol{\nabla}^{k_4} f,
}
where $\sum k_i=k$. Note that the momentum variables in front of the Riemann tensor, which arises as part of the $\boldsymbol \Gamma$ terms, can increase while appearing as coefficients of $\boldsymbol{\nabla}_i$ with $i\leq3$. In this case, the weights in the energy, appearing for lower numbers of derivatives, allow for these terms to be estimated by the energy. All the remaining terms arising from \eqref{dfiu} can be estimated similarly.
\end{proof}

\subsection{Estimating the energy momentum tensor}

\begin{lem}\label{lem-T-est}
Under smallness assumptions and $\ell\geq4$ the following estimates hold.
\eq{\alg{
\Abk{\rho(f)}\ell+\Abk{\jmath(f)}\ell&\leq C\cdot \ev \ell 3(f)\\
\Abk{\underline \eta(f)}\ell+\Abk{\underline T(f)}\ell&\leq C \cdot \ev \ell 4(f)\\
\Abk{S(f)}\ell&\leq C \left( \br \ell(f)+\ab\tau^2\ev\ell 4(f)\right)
}}

\end{lem}

\begin{proof}
We begin by estimating an integral of the form $\int_{TM}F\cdot G(\abg{p})\mu_{TM}$ for functions $G$, $F$ on $TM$ to explain the number of momentum weights. Let $\mu\geq 2$, then 
\eq{\alg{
\int_{TM}F\cdot G\mu_{TM}&=\iint  F\cdot G \sg dp \sg dx \\
&\leq \int \left(\int F^2 G^2 \overline p^{2\mu}\sg dp\right)^{1/2}\cdot\left(\int \overline p^{-2\mu}\sg dp\right)^{1/2}\sg dx\\
&\leq  \left(\iint F^2 G^2 \overline p^{2\mu}\sg dp\sg dx\right)^{1/2}\cdot\left(\underbrace{\iint \overline p^{-2\mu}\sg dp\sg dx}_{\leq C=C(\vol_g(M),\mu)}\right)^{1/2}.
}}
Depending on the additional momentum factors in $G$, which are of order one for $\rho$ and $\jmath$ and two for the other quantities, this explains the order of weights, necessary in the energies.
In the above computation $F$ represents the term where derivatives have acted on the distribution function and other quantities in the matter variables. We discuss how to estimate these terms in the following.\\
Covariant derivatives of matter quantities correspond to horizontal derivatives under the momentum-integral by the following identity,
\eq{\alg{
\nabla_a \int f \sg dp&=\int \partial_a f\sg dp +f \partial_a \sg dp\\
&=\int (\p a f+\Gamma_{ia}^i f)\sg dp\\
&=\int (\p a f-p^i\Gamma^e_{ia}\B_e f)\sg dp=\int \mbf A_a f \sg dp.
}}
Similar identities hold, for $f$ replaced by $f p^0$ etc. and for higher derivatives. For higher derivatives, we obtain not the full covariant derivative of the Sasaki metric. The additional terms arising from the Riemann tensor in \eqref{co-coeff} can however be added and substracted where the additional terms are lower order and due to the smallness condition, can be absorbed into the constants. Finally, if the horizontal derivative hits the momentum variables such as $\underline p$ or $p^0$ we use the formulae
\eq{\alg{
\A_a (\widehat X_i p^i)&=p^i\nabla_a\widehat X_i\\
\A_a(\widehat p)&=\frac{1}{2\widehat p}\left(\tau^2\widehat X_jp^jp^i\nabla_a \widehat X_i+\partial_a(\abg{\widehat X}^2)(1+\tau^2\abg{p}^2)\right)
}}
and estimate the arising shift vector terms using the smallness condition by the constants.
\end{proof}

\section{Control of the momentum support}
Using the characteristic system associated with the rescaled transport equation we derive an estimate on the supremum of the outer radius of the support of the distribution in momentum space. \\
The characteristic system corresponding to the rescaled transport equation \eqref{eq-transport} reads 
\eq{
\alg{
\frac{dX^a}{dT}&=-\tau \frac{p^a}{p^0} \\\qquad \frac{dP^a}{dT}&=\left(-\tau\p TX^a+\tau^{-1}\Gamma^a\right)p^0-2p^a+2\Gamma^a_ip^i,\\
&\quad\,+\left(\frac1N(\Si_{ij}+\frac13g_{ij})X^a+\tau\Gamma^a_{ij}\right)\frac{p^ip^j}{p^0}.
}}

We define the auxiliary quantity
\eq{
\mbf G(T,x,p)\equiv \abg{p}^2.
}
Using the characteristic system we compute the derivative of $\mbf G$ along a given characteristic. This yields
\eq{\alg{
\frac{d\mbf G}{dT}&=\ab{p}_{\dot g}^2+2\left(-\tau\langle p,\p TX\rangle_g+\tau^{-1}\langle\Gamma^*,p\rangle_g\right)p^0\\
&\quad\, +4p_i\Gamma^i_jp^j+\frac2N\left(\Si_{ij}\frac{p^ip^j}{p^0}\langle X,p\rangle_g+\frac13\frac{\abg p^2}{p^0}\langle X,p\rangle_g\right),
}} 
where it is important to recall that the rescaled momentum variables are time-dependent. Invoking the corresponding estimates for $p^0$ and $\abg{p}(p^{0})^{-1}$ in \eqref{p-est1}, \eqref{p-est2} we deduce the following estimate.

\begin{lem}
Under smallness assumptions, the following estimate holds for any characteristic.
\eq{\alg{
\Big|\frac{d\mbf G}{dT}\Big|&\leq C\left(\ab{\dot g}_g+\ab\tau^2\abg{\p TX}+\ab{\Gamma^*}+\ab{\Gamma^*_*}+\ab\tau^{-1}\abg\Si\abg X+\ab{\tau}^{-1}\abg X\right)\mbf G\\
&\quad\,+ C\left(\ab\tau\abg{\p T X}+\ab{\tau}^{-1}\abg{\Gamma^*}\right)\sqrt{\mbf G}
}}
\end{lem}

We define the supremum of the values of $\mbf G$ in the support of $f$ at a fixed time $T$ by 
\eq{
\boldsymbol{\mcr G}[T]\equiv \sup\{ \sqrt{\mbf G(T,x,p)}\, |\, (x,p)\in \supp f(T,.\,,.) \}.
}
From the estimate for individual characteristics above, we derive an estimate for $\mcr G$, which serves as a bound for momenta in the support of the distribution function.

\begin{prop}\label{prop-supp-est}
Under smallness assumptions we obtain

\eq{\alg{
\mcr G\Big|_{T}&\leq\left(\mcr G \Big|_{T_0}+C\int_{T_0}^T(e^{-s}\Abk{\p TX}2+e^s\Abk{\Gamma^*}2)ds\right)\\
&\qquad\times\exp\Big[C\int_{T_0}^T\big(e^s\Abk{X}3+\Abk{\Si}2+\Abk{N-3}3\\
&\qquad\qquad\qquad\qquad+e^{-2s}\Abk{\p TX}2+\Abk{\Gamma^*}2+\Abk{\Gamma^*_*}2\big)ds \Big].
}}

\end{prop}

\begin{proof}
For any characteristic in the support of $f$ we obtain an inequality of the form
\eq{\alg{
\frac{d}{dT}\sqrt{\mbf G}&\leq C\left(\ab\tau^{-1}\Abk{X}3+\Abk{\Si}2+\Abk{N-3}3+\ab{\tau}^2\Abk{\p TX}2+\Abk{\Gamma^*}2+\Abk{\Gamma^*_*}2\right) \sqrt{\mbf G}\\
&\quad+C(\ab\tau\Abk{\p TX}2+\ab{\tau}^{-1}\Abk{\Gamma^*}2).
}}
Gronwall's lemma implies
\eq{\alg{
\sqrt{\mbf G}\Big|_{T}&\leq\left(\sqrt{\mbf G}\Big|_{T_0}+C\int_{T_0}^T(e^{-s}\Abk{\p TX}2+e^s\Abk{\Gamma^*}2)ds\right)\\
&\qquad\times \exp\Big[C\int_{T_0}^T\Big(e^s\Abk{X}3+\Abk{\Si}2+\Abk{N-3}3\\
&\qquad\qquad\qquad\qquad+e^{-2s}\Abk{\p TX}2+\Abk{\Gamma^*}2+\Abk{\Gamma^*_*}2\Big)ds \Big].
}}
\end{proof}


\section{Energy estimates from the divergence identity}
The key quantity, which provides improved estimates for the energy density is the standard $L^2$-Sobolev energy for the rescaled energy density $\rho$ with respect to the dynamical metric $g$ on $(M,g)$. This energy reads
\eq{
\boldsymbol{\varrho}_\ell (\rho)\equiv \sqrt{\sum_{k\leq\ell}\int_M \abg{\na^k\rho}^2\mu_g}.
}

We derive the energy estimate for $\boldsymbol{\varrho}_\ell$ in the following. We denote
\eq{
\hp 0\equiv \p T + \mathcal{L}_X
}
as this combination of derivatives naturally appears when taking the time-derivative of norms, which are taken w.r.t.~the volume form $\mu_g$ (cf.~below).
A part of the divergence-identity for the energy momentum tensor with $0$-component reads
\eq{\label{sdfsd}
\hp 0 \rho = (3-N)\rho-\tau N^{-1}\na_a(N^2\jmath^a)+N\tau^2(\Si_{ab}+\frac13 g_{ab})T^{ab}
}
in its rescaled form (cf.~ \eqref{tdrhoeta}). Two identities relevant for the energy estimate in rescaled form are
\eq{
\hp 0 g^{ab}=-2N\Si^{ab}+2(1-N/3)g^{ab}.
}
and
\eq{
\p T \int_M u \mu_g= -\int_M (3-N)u\mu_g+ \int_M\widehat{\partial}_0(u)\mu_g 
}
for a function $u$ on $M$ (cf.~\cite{ChMo01}). Moreover,
\eq{\alg{
&[\widehat{\partial}_0, \nabla_{i}]\nabla_{j_1}\hdots\nabla_{j_m}u\\
&=- \sum_{a\leq m}\nabla_{j_1}\hdots\nabla_{j_{a-1}}\nabla_{b}\nabla_{j_{a+1}}\hdots\nabla_{j_m}u\cdot\left[\nabla_{i}(Nk^b_{j_a})+\nabla_{j_a}(Nk_i^b)-\nabla^b(Nk_{j_ai})\right]
}}
for a function $u$. This identity arises from the corresponding unrescaled one by multiplication with $-\tau$.

Next, we derive the standard energy estimate for this energy.

\begin{prop}\label{prop-rho-est}
Let $\ell\geq  4$
\eq{\alg{
\left|\p T\boldsymbol{\varrho}_{\ell}(f)\right|&\lesssim \left(\Abi{3-N}+{\Abi{N\Si}}+{\Abk{\na(Nk)}{\ell-2}}+{\Abk{3-N}{\ell}}\right)\cdot \boldsymbol{\varrho}_{\ell}(f)\\
&\qquad\quad+{\ab{\tau}^2\Abk{N(\Si_{ab}+\frac13g_{ab})T^{ab}}{\ell}+\ab{\tau}\Abk{N^{-1}\div(N^2\jmath)}{\ell}} 
}}
\end{prop}
\begin{proof}
We take the time derivative of one of the summands of the square of the energy, which takes the following form.
\eq{\alg{
&\p t \int_M\ab{\nabla^k(\rho)}^2_{g}\mu_{g}=- \int_M(3-N)\ab{\nabla^k(\rho)}^2_{g}\mu_{g}\\
&\quad+2\int_M(1-N/3)\ab{\nabla^k(\rho)}^2_{g}\mu_{g}\\
&\quad+2\sum_{i\leq k}\int_M N \Si_{cd}g^{ca_i}g^{db_i}g^{a_1b_1}\hdots g^{a_kb_k}\nabla_{a_1}\hdots\nabla_{a_k}(\rho) \cdot \nabla_{b_1}\hdots\nabla_{b_k}(\rho) \mu_{g}\\
&\quad+2\underbrace{\int_Mg^{a_1b_1}\hdots g^{a_kb_k}\nabla_{a_1}\hdots\nabla_{a_k}(\rho) \cdot \hp 0\left[\nabla_{b_1}\hdots\nabla_{b_k}(\rho)\right]\mu_{g}}_{\equiv\mathrm{I}}
}}
The first three terms on the right-hand side contribute to the first line of the estimate. We proceed with the evaluation of the term \textrm{I}.\\

Using the commutator formula above we obtain
\eq{\alg{
\mathrm{I}&= \underbrace{\int_Mg^{a_1b_1}\hdots g^{a_kb_k}\nabla_{a_1}\hdots\nabla_{a_k}(\rho) \cdot\nabla_{b_1}\hdots\nabla_{b_k}( \hp 0(\rho))\mu_{g}}_{\equiv\mathrm{II}}\\
&\quad+ \int_Mg^{a_1b_1}\hdots g^{a_kb_k}\nabla_{a_1}\hdots\nabla_{a_k}(\rho)\\
&\qquad \cdot 
\Big[\sum_{i\leq k-1}\sum_{i+1\leq j \leq k}\nabla_{b_1}\hdots \nabla_{b_{i-1}}\Big(\nabla_{b_{i+1}}\hdots \nabla_{b_{j-1}}\nabla_{c}\nabla_{b_{j+1}}\hdots\nabla_{b_{k}}(\rho)\cdot K^c_{b_{j}b_i}\Big)\Big]\mu_{g},
}}
where we use the notation
\eq{\alg{
K^a_{bc}&=\left[\nabla_{b}(Nk^a_{c})+\nabla_{c}(Nk_b^a)-\nabla^a(Nk_{cb})\right].
}}

The second term on the right-hand side can be estimated by terms of the form
\eq{
C \boldsymbol{\varrho}^2_{\ell}(f)\cdot\Abk{K}{\ell-2},
}
which yield the third term in the estimate. We continue with estimating the final term $\mathrm{II}$ using \eqref{sdfsd}.
\eq{\alg{
\ab{\mathrm{II}}&=\Big|\int_Mg^{a_1b_1}\hdots g^{a_kb_k}\nabla_{a_1}\hdots\nabla_{a_k}(\rho)\\ &\qquad\cdot\nabla_{b_1}\hdots\nabla_{b_k}\left( (3-N)\rho-\tau N^{-1}\na_a(N^2\jmath^a)+N\tau^2(\Si_{ab}+\frac13 g_{ab})T^{ab}\right)\mu_{g}\Big|\\
&\leq \boldsymbol{\varrho}_{\ell}(f)\Big( \Abk{3-N}{\ell}\cdot \boldsymbol{\varrho}_{\ell}(f)+\ab{\tau}\Abk{N^{-1}\div(N^2\jmath)}{\ell}\\
&\qquad\qquad\quad\,+\ab{\tau}^2\Abk{N(\Si_{ab}+\frac13g_{ab})T^{ab}}{\ell}\Big)
}}
\end{proof}

\section{Elliptic estimates}\label{sec : ell}
We derive in this section elliptic estimates on the lapse function, the shift vector and their respective time derivatives.
\begin{prop}\label{prop-ell-est}
Under smallness conditions, for the lapse function, a pointwise estimate of the form $0<N\leq 3$ holds and moreover the following two estimates.
\eq{\alg{
\Abk{N-3}\ell&\leq C\left(\Abk{\Si}{\ell-2}^2+\ab\tau\Abk{\rho}{\ell-2}+\tau^3 \Abk{\underline{\eta}}{\ell-2}\right)\\
\Abk{X}\ell&\leq C\left(\Abk{\Si}{\ell-2}^2+\Abk{g-\gamma}{\ell-1}^2+\ab\tau\Abk{\rho}{\ell-3}+\tau^3 \Abk{\underline{\eta}}{\ell-3}+\tau^2 \Abk{N\jmath}{\ell-2}\right)
}}
\end{prop}

\begin{proof}
The pointwise estimate for the lapse follows from the lapse equation and the maximum principle. The two following estimates are a straightforward consequence from elliptic regularity applied to the elliptic system for lapse and shift.
\end{proof}

Furthermore we require estimates for the time derivatives of of the lapse function and shift vector. These are given in the following lemma.

\begin{lem}
The following estimates hold under smallness conditions, for $T$ sufficiently large and  $\ell\geq 4$.
\eq{\alg{
\Abk{\p TN}\ell&\leq C\Big[\Abk{\widehat N}\ell+\Abk{X}{\ell+1}+\Abk{\Si}{\ell-1}^2+\Abk{g-\gamma}\ell^2+\tau\Abk{S}{\ell-2}\\
&\qquad\quad+\ab{\tau}\Abk{\rho}{\ell-1}+\ab\tau^3\Abk{\underline\eta}{\ell-2} +\ab\tau^2\Abk{\jmath}{\ell-1}+\ab\tau^3\Abk{\underline T}{\ell-1}\\
&\qquad\quad+\ab\tau^3\ev {\ell-1}{4}(f)\Big]\\
\Abk{\p TX}\ell&\leq C\Big[\Abk{X}{\ell+1}+\Abk{\Si}{\ell-1}^2+\Abk{g-\gamma}{\ell}^2+\Abk{\widehat N}{\ell}\\
&\qquad\quad+\ab{\tau}\Abk{\rho}{\ell-1}+\ab\tau^3\Abk{\underline\eta}{\ell-2}+\ab\tau^2\Abk{\jmath}{\ell-1}+\ab\tau^3\Abk{\underline T}{\ell-1}\\
&\qquad\quad+\ab\tau\Abk{S}{\ell-2}+\ab\tau^3\ev {\ell-1}{4}(f)\Big]
}}
\end{lem}

\begin{proof}
Both estimates follow from standard elliptic regularity estimates and the elliptic system for $(\p TN,\p TX)$, which is deduced from the elliptic system for $(N,X)$ by taking the derivative with respect to $\p T$. This system reads
\eq{\label{tdlapse}\alg{
\/\left(\Delta-\frac13\right)\p TN&= 2N\langle\nabla\nabla N,\Si\rangle-2\widehat N\Delta N+\langle\nabla\nabla N,\mcl L_Xg\rangle\\
&\quad+\left(2\nabla^k(N\Si_k^i)+\nabla^i(\widehat N)-\frac12\Delta X^i-\frac12\nabla^k\nabla^iX_k\right)\nabla_i N\\
&\quad+2N\Big(-2N\abg{\Si}^3+2\widehat N\abg{\Si}^2-2\langle\nabla X,\Si,\Si\rangle-2\abg{\Si}^2\\
&\qquad\qquad\,\,-N\langle\Si,\frac12\mcl L_{g,\ga}(g-\gamma)+J\rangle+\langle\Si,\nabla\nabla N\rangle+2N\abg\Si^3-\widehat N\abg\Si^2\\
&\qquad\qquad\,\,-2\langle\Si,\mcl L_Xg\rangle+8\pi\ab\tau\langle\Si,S\rangle\Big)\\
&\quad+N\Big(\p T(\ab{\tau}\rho)+\p T(\ab{\tau}^3\underline\eta)\Big)+\left(\abg{\Si}^2+\ab{\tau}\rho+\ab{\tau}^3\underline{\eta}\right)\p TN
}}

\eq{\alg{
\Delta (\p T X^i)+R_m^i(\p T X^m)&=-(\p TR^i_m)X^m-[\p T,\Delta]X^i\\
&\quad+2\nabla_j(\p T N)\Si^{ij}+2\nabla_j N(\p T\Si^{ij})-(\p Tg^{ik})\nabla_k\widehat N-\frac13g^{ik}\na_k(\p TN)\\
&\quad+2(\p TN)\ab\tau^{2}\jmath^b+2N\p T(\ab\tau^{2}\jmath^b)\\
&\quad- 2(\p TN)\Si^{mn}(\Gamma_{mn}^i-\widehat{\Gamma}_{mn}^i)- 2N(\p T\Si^{mn})(\Gamma_{mn}^i-\widehat{\Gamma}_{mn}^i)\\
&\quad- 2N\Si^{mn}\p T\Gamma_{mn}^i+ (\p Tg^{mk}g^{nl})\nabla_kX_l(\Gamma_{mn}^i-\widehat{\Gamma}_{mn}^i)\\
&\quad+ \nabla^m(\p TX^n)(\Gamma_{mn}^i-\widehat{\Gamma}_{mn}^i)
+ \nabla^mX^n\p T\Gamma_{mn}^i.
}}
Here we use $\langle.,.,.\rangle$ to denote any suitable contraction of a number of tensor fields, where the specific structure of indices does not matter. Due to the time derivative of $\underline{\eta}$ and the terms containing $\p TN$ explicitly in the equation for $\p T X$ we do the estimates in two steps. Note furthermore, that we do not aim at the sharpest possible estimates and allow rather rough but brief expressions where we absorb many terms into the constants.\\

From elliptic regularity and equation \eqref{tdlapse} we obtain
\eq{\label{lasd}\alg{
\Abk {\p T N}\ell&\leq C\Big[\Abk{\widehat N}\ell+\Abk{X}\ell+\Abk{\Si}{\ell-1}^2+\Abk{g-\gamma}\ell^2+\tau\Abk{S}{\ell-2}\\
&\qquad\quad+\ab{\tau}\Abk{\rho}{\ell-2}+\ab\tau^3\Abk{\underline\eta}{\ell-2}\\
&\qquad\quad +\ab\tau \Abk{\p T\rho}{\ell-2}+\ab\tau^3\Abk{\p T\underline\eta}{\ell-2}\\
&\qquad\quad+\left(\Abk{\Si}{\ell-2}^2+\ab\tau\Abk{\rho}{\ell-2}+\ab\tau^3\Abk{\underline\eta}{\ell-2}\right)\Abk{\p TN}{\ell-2}\Big]
}}
Using the smallness we can absorb the last line of the previous equation into the left-hand side and obtain a formally identical estimate where the last line is not present. The term including the time derivative of $\rho$ is treated using the evolution equation \eqref{tdrhoeta}. This yields
\eq{
\ab{\tau}\Abk{\p T\rho}{\ell-2}\leq C\ab\tau\left(\Abk{\rho}{\ell-1}+\ab\tau\Abk{\jmath}{\ell-1}+\ab\tau^2\Abk{\underline T}{\ell-2}\right).
}
Now, we estimate the remaining term using the corresponding formula \eqref{sdfjdsaa}. Invoking the smallness assumption and the fact that when taking derivatives of the explicit function of the momentum only yields terms with an additional smallness factor,  then reduces the number of relevant terms to an estimate of the following form.
\eq{\label{eraguiu}\alg{
\ab\tau^3\Abk{\p T\underline\eta}{\ell-2}&\leq \ab\tau^3C\Big[\ab\tau\ev {\ell-1}{\mu+3}(f)+\ev{\ell-2}{\mu+3}(f)\\
&\quad\,\qquad\quad+\Abk{\p TN}{\ell-2}\ab\cdot\ev{\ell-2}{\mu+3}(f)\\
&\quad\,\qquad\quad+\Abk{\p T X}{\ell-2}\ab{\tau}^{-1}\cdot\ev{\ell-2}{\mu+3}(f)\Big]
}}
The term in the second line can be absorbed in the constant in estimate \eqref{lasd} by the largeness of $T$. Before concluding the estimate for $\p TN$ we require the estimate for $\p TX$ to replace the corresponding terms in the previous estimate. We therefore turn to the equation for $\p TX$ and apply elliptic regularity which yields the following first estimate, where we again absorb several terms in the constant due to the smallness criterion.
\eq{\label{eragui}\alg{
\Abk{\p TX}\ell&\leq C\Big[\Abk{X}{\ell+1}+\Abk{\p TN}{\ell-1}+\Abk{\widehat N}{\ell-2}
+\ab\tau^2\Abk{\jmath}{\ell-2}+\ab\tau^2\Abk{\p T\jmath}{\ell-2}\\
&\qquad\quad+\Abk{g-\gamma}{\ell-1}\Abk{\dot \Si}{\ell-2}+\Abk{\Si}{\ell-2}\Abk{\dot \Gamma}{\ell-2}+\Abk{X}{\ell-1}\Abk{g-\gamma}{\ell-1}\\
&\qquad\quad+\Abk{X}{\ell-1}\Abk{\dot\Gamma}{\ell-2}+\Abk{\p TX}{\ell-1}\Abk{g-\gamma}{\ell-1}\Big]
}}
The last term on the right-hand side can be absorbed into the constant by the smallness assumption. The term containing the time derivative of $\jmath$ can be estimated using \eqref{tdrhoeta} by
\eq{\alg{
\ab\tau^2\Abk{\jmath}{\ell-2}&\leq \ab\tau^2C\Big[\Abk{\jmath}{\ell-1}+\ab\tau\Abk{\underline T}{\ell-1}+\ab{
\tau}^{-1}\Abk{\rho}{\ell-2}\Abk{\widehat N}{\ell-1}\Big].
}}
At this point the estimate for $\p TX$ is not complete, since there are still $\p TN$ terms on the right-hand side. We return to the estimate for $\p TN$ and absorb the corresponding terms in the estimate and then finish the estimate for $\p TX$.\\
Plugging \eqref{eragui} without the last term on the right-hand side into \eqref{eraguiu} and the resulting estimate into \eqref{lasd}, without the last line on the right-hand side, we observe that every $\p TN$ term on the right-hand side comes with a $\ab{\tau}^{3}$ 
and consequently can be absorbed into the constant. This proves the estimate for $\p TN$, which in particular is independent of $\p TX$. Then, in turn, plugging the final estimate for $\p TN$ into the estimate for $\p TX$ and simplifying the estimates with respect to the smallness criteria finishes the proof.
\end{proof}


\section{Energy estimate -- geometry}
\subsection{Decomposing the evolution equations}
We decompose the evolution equations into their principle parts and higher order terms, which are eventually treated as bulk terms.\\
\\
The evolution equations can be rewritten to the following system.
\eq{\alg{
\p T(g-\ga)&=2N\Si + \mcl F_{g-\gamma}\\
\p T6\Si&=-2\cdot 6\Si-9\frac N3\mcl L_{g,\ga}(g-\gamma) + 6 N\ab{\tau} S -X^i\nabla[\gamma]_i6\Si
+\mcl F_{\Si},
}}
where the bulk terms obey estimates of the form
\eq{\alg{
\Abk{\mcl F_{g-\gamma}}s&\leq C\left( \Abk{N-3}s+\Abk{X}{s+1}\right)\\
\Abk{\mcl F_\Si}{s-1}&\leq C\left(\Abk{g-\gamma}s^2+\Abk{N-3}{s+1}+\Abk{\Si}{s-1}^2+\Abk{X}{s}\right)
}}
under the assumption that $\Abk{\Si}{s-1}^2+\Abk{g-\ga}s^2<\varepsilon$ for $\varepsilon$ sufficiently small. Using the elliptic estimates for lapse and shift we obtain the following estimates for the bulk terms.
\eq{\alg{
\Abk{\mcl F_{g-\gamma}}s&\leq C\left( \Abk{\Si}{s-1}^2+\ab{\tau}\Abk{\rho}{s-2}+\ab{\tau}^3\Abk{\underline\eta}{s-2}+   \Abk{g-\gamma}{s}^2 +\ab{\tau}^2\Abk{N\jmath}{s-1}\right)\\
\Abk{\mcl F_\Si}{s-1}&\leq C\left(\Abk{g-\gamma}s^2+\ab{\tau}\Abk{\rho}{s-1}+\ab{\tau}^3\Abk{\underline\eta}{s-1}+\Abk{\Si}{s-1}^2+\ab{\tau}^2\Abk{N\jmath}{s-2}\right)
}}

\subsection{Energy}
We define the energy for the tracefree part of the second fundamental form and the metric perturbation below. The choice is identical to the vacuum case considered in \cite{AnMo11} and we briefly recall the relevant aspects and point out the improvements in 3+1 dimensions compared to the higher dimensional case. The definition of the energies, which include a correction factor to obtain a suitable decay estimate, depends on the lowest eigenvalue of the Einstein operator corresponding to the specific Einstein metric, $\la_0$. Due to the lower bound \eqref{ev-est} we only distinguish between two cases here. We define the correction constant $\al=\al(\la_0,\delta_\alpha)$ by
\eq{
\al=
\begin{cases}
1& \la_0>1/9\\
1-\delta_\alpha& \la_0=1/9,
\end{cases}
}
where $\delta_\alpha=\sqrt{1-9(\la_0-\varepsilon')}$ with $1>>\varepsilon'>0$ remains a variable to be determined in the course of the argument to follow. By fixing $\varepsilon'$ once and for all, $\delta_\alpha$ can be made suitable small when necessary.\\
The corresponding correction constant, relevant for defining the corrected energies is defined by
\eq{
c_E=\begin{cases}
1& \la_0>1/9\\
9(\la_0-\varepsilon')& \la_0=1/9.
\end{cases}
}
We are now ready to define the energy for the geometric perturbation. For $m\geq 1$ let
\eq{\alg{
\mathcal{E}_{(m)}&=\frac12\int_M\langle 6\Si,\mcl L_{g,\ga}^{m-1} 6\Si\rangle\mu_g+\frac92\int_M \langle (g-\gamma),\mcl L_{g,\ga}^{m}(g-\gamma)\rangle\mu_g\\
\Gamma_{(m)}&=\int_M \langle 6\Si,\mcl L_{g,\ga}^{m-1}(g-\gamma)\rangle\mu_g.
}}

Then, the corrected energy for the geometric perturbation is defined by
\eq{
E_s=\sum_{1\leq m\leq s} \mcl E_{(m)}+c_E\Gamma_{(m)}.
}
Under the imposed conditions, the energy is coercive.

\begin{lem}
There exists a $\delta>0$ and a constant $C>0$ such that for $\delta$-small data $(g,\Si,f)$ the inequality
\eq{
\Abk{g-\gamma}6^2+\Abk{\Si}5^2\leq C E_s(g,\Si)
}
holds.
\end{lem}

\begin{proof}
The proof is analogous to the corresponding Lemma 7.2 in \cite{AnMo11}. The difference consists in the fact that in the 3+1 dimensional setting here, the kernel of the Einstein operator consists only of the zero-tensor (cf.~Corollary \ref{cor-kernel}). This implies that the projection operator necessary in Lemma 7.2 \cite{AnMo11} is not necessary in the present case.
\end{proof}

The energy estimate for the corrected energy is given in the following.
\begin{lem}\label{lem-geom-en-est}
Under a smallness assumption on $E_s$ we have
\eq{\label{en-est-geom}\alg{
\p T E_s &\leq -2\al E_s+6 E_s^{1/2}\ab{\tau}\Abk{NS}{s-1}\\
&\quad+ CE_s^{3/2} +C E_s^{1/2}\left(\ab{\tau}\Abk{\rho}{s-1}+\ab{\tau}^3\Abk{\underline\eta}{s-1}+\ab{\tau}^2\Abk{N\jmath}{s-2}\right)
}}
\end{lem}

\begin{proof}
The proof is analogous to the proof of Lemma 7.6 in \cite{AnMo11}. The only difference results from the additional matter term in the evolution equation for $\Si$. As a direct consequence of the equation, this yields terms of the types
\eq{\alg{
&\int_M\langle N\tau S,\mathcal{L}_{g,\gamma}^{m-1}\Si\rangle\mu_g+\int_M\langle \Si,\mathcal{L}_{g,\gamma}(N\tau S)\rangle\mu_g,\\
&\int_M\langle N\tau S,\mcl L^{m-1}_{g,\gamma}(g-\gamma)\rangle\mu_g,
}}
which can straightforwardly be estimated by
\eq{
\ab{\tau} \Abk{NS}{s-1}\sqrt{E_s},
}
yielding the claim. \end{proof}


\section{Total energy estimate}
With the individual energy estimates for geometry and matter variables at hand these require to be synchronized in view of their different decay inducing terms. For this purpose we define a total energy with explicit weight functions in time and bound all elliptic variables in terms of this energy. We then derive energy estimates under the smallness assumption on $\boldsymbol{\rho}_4(f)$, $\mcr G$ and the total energy which are the key estimates to establish the global existence result further below.

\subsection{Total energy}
We define the total energy including the matter energy and the energy for the metric perturbation.

\begin{dfn}
\eq{
\mbf E_{\mathrm{tot}}(\Si,g-\gamma,f)\equiv e^{(1+\delta_E)T}E_6(g-\gamma,\Si)+ e^{-\delta_{\mcr E}\cdot T}\ev54^2(f),
}
where $\delta_E+\delta_{\mcr E}<1$ and $\delta_E<1/2$, $\delta_{\mcr E}>1/2$.

\end{dfn}
 We choose now all auxiliary constants in the following way. For a given $\boldsymbol{\varepsilon}_{\mathrm{decay}}<1$ we choose positive constants $(\delta_\alpha,\delta_E,\delta_{\mcr E},\varepsilon_{\mathrm{tot}})$ such that
\eq{\label{cond-aux-const}\alg{
1-2\delta_{\alpha}-\delta_E-\varepsilon_{\mathrm{tot}}>1-\boldsymbol{\varepsilon}_{\mathrm{decay}}\\
\delta_{\mcr E}-\varepsilon_{\mathrm{tot}}>1-\boldsymbol{\varepsilon}_{\mathrm{decay}}
}}
hold. For small $\boldsymbol{\varepsilon}_{\mathrm{decay}}$ this is achieved, when $\delta_{\mcr E}$ is almost one and $\delta_E$ is sufficiently small relative to $\delta_{\mcr E}$ such that $\delta_E+\delta_{\mcr E}<1$ holds.
We define a uniform constant $\overline C$, that bounds all constants $C$ in previous estimates from above by
\eq{
10 \cdot C^{3}\leq \overline C.
}

\subsection{Preparations}
We gather now a number of simplifying lemmas to reduce the length of the final energy estimate. We express in the following all relevant norms in terms of the energies $E_s$, $\ev 5 4(f)$, $\br 4(f)$ and $\mcr G$.
For the norms appearing in the energy estimate for the $L^2$-energies we have

\begin{lem}\label{lem-short-est}
Under suitable smallness assumptions the following estimates hold.
\eq{\alg{
&\Abk{3-N}6\leq C\left(e^{-(1+\delta_E)T}\et+e^{-T}\br 4(f)+e^{-(3-\delta_{\mcr E}/2)T}\sqrt{\et}\right)\\
\Abk{X}6&+\Abk{\overset{\circ}{\Gamma^*}}5+\Abk{\Gamma^*_*}5+\Abk{\p TX}5+\Abk{\p TN}5\\
&\leq C\left(e^{-(1+\delta_E)T}\et+e^{-T}\br 4(f)+e^{-(3-\delta_{\mcr E}/2)T}\sqrt{\et}+e^{-(2-\delta_{\mcr E}/2)T}\sqrt{\et}\right)
}
}
In total,
\eq{\label{comb-ell-est}\alg{
\Abk{3-N}6&+\Abk{X}6+\Abk{\overset{\circ}{\Gamma^*}}5+\Abk{\p TX}5+\Abk{\p TN}5+
\Abk{\Gamma^*_*}5\\
&\leq C\left(e^{-(1+\delta_E)T}\et+e^{-T}\br 4(f)+e^{-(2-\delta_{\mcr E}/2)T}\sqrt{\et}\right).
}}
\end{lem}

\subsection{Estimates for $\br 4(f)$}

We begin with an estimate for the auxiliary energy of the energy density.

\begin{lem}\label{lem-br-est}
For $\delta$-small data with $\delta$ sufficiently small, the folllwing estimate holds.
\eq{
\br 4(f)\Big|_{T}\leq \left(\br 4(f)\Big|_{T_0}+C\int_{T_0}^Te^{-(1-\delta_E/2)\cdot s}\,\sqrt{\mbf E_{\mathrm{tot}}\Big|_s} ds\right)\cdot\exp\left[C\int_{T_0}^Te^{- s/2}\,\sqrt{\mbf E_{\mathrm{tot}}\Big|_s} ds\right]
}
\end{lem}
\begin{proof}
From Proposition \ref{prop-rho-est}, using Lemma \ref{lem-T-est} and Proposition \ref{prop-ell-est}, we obtain

\eq{\alg{
\p T\br 4(f)\leq C\left(\Abk{\Si}3+\ab\tau \ev 4 4(f)\right)\cdot \br 4(f)+\ab\tau\ev 54(f).
}}
Estimating by the total energy and integrating yields
\eq{
\br 4(f)\Big|_{T}\leq \br 4(f)\Big|_{T_0}+C\int_{T_0}^Te^{-(1-\delta_E/2)\cdot s}\,\sqrt{\mbf E_{\mathrm{tot}}\Big|_s} ds+C\int_{T_0}^Te^{- s/2}\,\sqrt{\mbf E_{\mathrm{tot}}\Big|_s} \cdot\br 4(f)\Big|_s ds.
}
Then, Gronwall's lemma yields the claim. 
\end{proof}

\subsection{Estimate on $\mcr G$}
For the bound on the support of the momentum variables we obtain the following estimate.

\begin{lem}\label{lem-supp-est}
For $T_0>1$ and under the $\delta$-smallness assumption
for $\delta$ sufficiently small, the following estimate holds.
\eq{\alg{
\mcr G\Big|_{T}&\leq\left(\mcr G \Big|_{T_0}+C\int_{T_0}^T\left(e^{-\delta_E\cdot s}\et+\br 4(f)+e^{-(1-\delta_{\mcr E}/2)s}\sqrt{\et}\right)
ds\right)\\
&\qquad\times \exp\Big[C\int_{T_0}^T\Big(e^{-\delta_E\cdot s}\et+\br 4(f)+e^{-(1-\delta_{\mcr E}/2)s}\sqrt{\et}\Big)ds \Big]
}}

\end{lem}

\begin{proof}
The estimate follows directly from Proposition \ref{prop-supp-est} in combination with Lemma \ref{lem-short-est}.\end{proof}


\subsection{Estimate -- total energy}
We proceed with an estimate on the total energy under a smallness assumption on the auxiliary energy.
\begin{prop}\label{prp-et-est}
Under the assumption of $\delta$-smallness and the conditions 
\eq{\label{small-1}
\overline C\br{4}(f)\leq  \varepsilon_{\mathrm{tot}}/3
} 
and 
\eq{\label{small-2}
\overline C \ab\tau\mcr G\leq\varepsilon_{\mathrm{tot}}/3
}
and for $T$ sufficiently large to assure
\eq{\label{large-time-cond}
\overline C e^{-(1-\tfrac12(1+\delta_E+\delta_{\mcr E}))T}< \varepsilon_{\mathrm{tot}}/3
}
the estimate
\eq{
\p T \mbf E_{\mathrm{tot}}\leq - \big[1-\boldsymbol{\varepsilon}_{\mathrm{decay}}\big]\mbf E_{\mathrm{tot}}+C\mbf E_{\mathrm{tot}}^{3/2}
}
holds.
\end{prop}

\begin{proof}
Taking the time derivative of the total energy, using the estimate for the energy for the perturbation of the geometry, Lemma \ref{lem-geom-en-est}, and the estimate for the $L^2$-Sobolev energy of the distribution function, 
Proposition \ref{prp-l2-est}, we obtain
\eq{
\alg{
\p T \mbf E_{\mathrm{tot}}(\Si,g-\ga,f)&\leq \underbrace{-(2\al-1-\delta_E) e^{(1+\delta_E)T}E_6}_{(1.1)}\\
&\quad+\underbrace{C (e^{(1+\delta_E)T}E_6)^{1/2}e^{-(1-\tfrac12(1+\delta_E+\delta_{\mcr E}))T}e^{-\delta_{\mcr E}T/2}\ev54(f)}_{(1.2)}\\
&\quad+ \underbrace{Ce^{(1+\delta_E)T}E_6^{3/2}}_{(1.3)}\\
&\quad\underbrace{-\delta_{\mcr E} e^{-\delta_{\mcr E}T}\ev54^2(f)}_{(2.1)}+\underbrace{C^2\br 4(f)e^{-\delta_{\mcr E}T}\ev 54^2(f)}_{(2.2)}\\
&\quad+\underbrace{C \ab\tau \mcr G \ev 54(f)^2}_{(2.3)}+\underbrace{C^2\mbf E_{\mathrm{tot}}^{3/2}}_{(2.4)}.
}}
The terms resulting from the energy estimate for $E_6$ are denoted by numbers $(1.i)$. The term $(1.1)$ results from the decay inducing term in the estimate \eqref{en-est-geom} and the time derivative of the time-weight function. The term $(1.2)$ results from any matter term in the estimate \eqref{en-est-geom}, where we have to estimate by the $L^2$-norm since the regularity is up to the order $s-1=5$. Note that the time-weight function is distributed to re-obtain the properly weighted energies as they appear in the total energy. Finally, term $(1.3)$ results from the higher order term.\\
The terms resulting from the energy estimate for $\mcr E_{5,4}(f)$ are denoted by numbers $(2.i)$. Term $(2.1)$ results from the time derivative of the time-weight function. Term $(2.2)$ bounds all terms from estimate \eqref{en-est-f-prp}, which result from the term $\tau^{-1}N^{-1}\Gamma^*$, which is estimated using \eqref{comb-ell-est} where only the term with $\boldsymbol{\rho}_4(f)$ is considered, all other terms are of higher order in energy and are absorbed into the term $(2.4)$ except for the term $\tau\mcr G$, which is estimated by 
$(2.3)$.\\

Using the smallness conditions appropriately, the previous estimate reduces to

\eq{
\alg{
&\p T \mbf E_{\mathrm{tot}}(\Si,g-\ga,f)\leq-(2\al-1-\delta_E) e^{(1+\delta_E)T}E_6-\delta_{\mcr E}e^{-\delta_{\mcr E}T}\ev 54^2(f)+\varepsilon_{\mathrm{tot}} \,\mbf E_{\mathrm{tot}}+ C\mbf E_{\mathrm{tot}}^{3/2}.
}}
Here, terms $(1.1)$ and $(2.1)$ appear as before and provide 
decay inducing terms. Terms $(1.3)$ and $(2.4)$ are absorbed in the higher order term. Invoking smallness conditions \eqref{small-1}, \eqref{small-2} and \eqref{large-time-cond} allows us to bound the sum of terms $(1.2)$, $(2.2)$ and $(2.3)$ by $\varepsilon_{\mathrm{tot}}\mbf E_{\mathrm{tot}}$.

This yields
\eq{
\p T \mbf E_{\mathrm{tot}}\leq - \big[(\min\{2\al-1-\delta_E,\delta_{\mcr E}\}-\varepsilon_{\mathrm{tot}})\mbf E_{\mathrm{tot}}\big]+C\mbf E_{\mathrm{tot}}^{3/2},
}
which under the conditions \eqref{cond-aux-const} on the auxiliary constants yields the claim.
\end{proof}

\section{Global existence and completeness}
In this final section we present the proof of Theorem 1 based on the estimates in the previous sections.

\subsection{Preliminaries}
We consider initial data at time $T_0$, which is close to the induced data of the Milne model at $T=T_0$. The data is not necessarily CMC initial data. We argue below why it is sufficient to consider only CMC initial data and consider this case for now. The existence of a local-in-time solution for CMC initial data close to the Milne geometry has been developed in \cite{Fa15} and we adapt the local-existence theory therein to our present notation and variables.\\
The local existence theorem (Theorem 4.2, \cite{Fa15}) assures existence of a unique local solutions for initial data $(g_0,k_0,f_0)\in H^6\times H^5\times H_{\mathrm{Vl},3}^5$, which is the regularity assumed in the present case. Moreover, this solution is depending on the initial data in a continuous sense, which allows to increase $T_0$ suitably and assume smallness at the increased $T_0$ without loss of generality. We denote the smallness parameter according to which we express smallness of the initial data in the sense of $\mcr B_{\varepsilon_0}^{6,5,5}$ by $\varepsilon_0$. To establish global existence we require a continuation criterion analogous to Theorem 8.1 in \cite{Fa15}. It is important to specify this to our present situation where we consider the rescaled system in 3+1-dimensions. If we replace the non-rescaled system in \cite{Fa15} by the rescaled equations \eqref{lapse} - \eqref{ev-k}, the smallness, which has to be assured to continue the solution translates to
\eq{
Q_{\mathrm{cont}}=\Abk{g-\gamma}5+\Abk{\Si}4+\ab{\tau}|\!|\!|f|\!|\!|_{4,3}+\Abk{N-3}5+\Abk{X}5+\Abk{\dot N}4+\Abk{\dot X}4<\varepsilon_{\mathrm{loc}},
}
for a fixed $\varepsilon_{\mathrm{loc}}>0$. This means, either the maximal interval of existence is infinite or the bound above is attained as this time is approached. In particular, starting with sufficiently small initial data, if this smallness persists throughout the evolution, global existence is automatically assured. This persistence is shown for the initial data we consider, which according to the previous discussion guarantees existence of the solution.

\subsection{Existence of a CMC surface}
Considering sufficiently small initial data which is not necessarily CMC, the maximal globally hyperbolic development under the Einstein-Vlasov system is, locally in time, as close to the background geometry as desired in a suitable regularity \cite{Ri13}. The existence of a CMC surface in such a spacetime can be shown along the lines of the corresponding argument in the vacuum case presented for instance in \cite{FK15}.

\subsection{Guaranteeing the smallness condition on an open interval}

From the local Cauchy stability by choosing  
the initial data sufficiently small we can assure existence of the solution up to $T_0$ and smallness at $T_0$ such that condition
\eqref{large-time-cond} holds at $T_0$.
We choose the new initial data at $T_0$ small such that
\eq{
\et \Big|_{T_0}+\br4(f)\Big|_{T_0}+\mcr G\Big|_{T_0}\leq \varepsilon_0.
}
Since all estimates are uniform in the sense that they do not depend on the smallness of the initial data once $\varepsilon_0$ is chosen sufficiently small, we can further decrease $\varepsilon_0$ in the course of the argument. The same holds for increasing $T_0$.\\

We choose $\varepsilon_0$ sufficiently small to assure that conditions \eqref{small-1} and \eqref{small-2} hold at $T_0$. We now define 
\eq{\label{wefoh}\alg{
T_{+}\equiv \sup\Big\{T>T_0 \Big| &\mbox{ The solution exists, is $\delta$-small}\\
&\mbox{ and conditions \eqref{small-1} and \eqref{small-2} hold on $[T_0,T)$.}\Big\}.
}}
By local existence $T_+>T_0$ exists. Note that the condition \eqref{large-time-cond} holds automatically at later times.
\subsection{Improving the bootstrap conditions -- Global existence}
We show in the following that if $\varepsilon_0>0$ is sufficiently small then $T_+=\infty$.\\

Due to the validity of conditions \eqref{small-1}, \eqref{small-2} and \eqref{large-time-cond} Proposition \ref{prp-et-est} holds on $(T_0,T_+)$, which yields 

\eq{
\frac{d\sqrt{\et}}{dT}\leq-\frac12(1-\boldsymbol{\varepsilon}_{\mathrm{decay}})\sqrt{\et}+\overline{C}\cdot \et
}
and in turn

\eq{\alg{
\sqrt{\et}\Big|_T&\leq\frac12\frac{1-\boldsymbol\varepsilon_{\mathrm{decay}}}{\overline C+ e^{(1-\boldsymbol\varepsilon_{\mathrm{decay}})/2\cdot(T-T_0)}\left((1-\boldsymbol\varepsilon_{\mathrm{decay}})/2\sqrt{\et\big|_{T_0}}^{-1}-\overline C\right)}\\
&\leq 2 \sqrt{\et\Big|_{T_0}} e^{-(1-\boldsymbol\varepsilon_{\mathrm{decay}})/2\cdot(T-T_0)},
}}
where in the second inequality we have further decreased $\varepsilon_0$ to assure $(1-\boldsymbol\varepsilon_{\mathrm{decay}})/2-\varepsilon_0\overline C>(1-\boldsymbol\varepsilon_{\mathrm{decay}})/4$ and $\varepsilon_0<\overline C^{-1}(1-\boldsymbol\varepsilon_{\mathrm{decay}})/2$.\\

Using the previous decay result in combination with Lemma \ref{lem-br-est} yields
\eq{\alg{
\br 4(f)\Big|_{T}&\leq \left(\br 4(f)\Big|_{T_0}+C\int_{T_0}^Te^{-(1-\delta_E/2)\cdot s}\,\sqrt{\mbf E_{\mathrm{tot}}\Big|_s} ds\right)\cdot\exp\left[C\int_{T_0}^Te^{- s/2}\,\sqrt{\mbf E_{\mathrm{tot}}\Big|_s} ds\right]\\
&\leq (\varepsilon_0+C'\varepsilon_0)\cdot\exp(C'\sqrt{\varepsilon_0})
}}
Choosing now $\varepsilon_0$ sufficiently small, this implies \eqref{small-1} with a strict inequality.\\

Finally, we invoke Lemma \ref{lem-supp-est}, which, using the previous estimates, takes the form
\eq{\alg{
\mcr G\Big|_{T}&\leq\left(\mcr G \Big|_{T_0}+C\int_{T_0}^T\left(e^{-\delta_E\cdot s}\et+\br 4(f)+e^{-(1-\delta_{\mcr E}/2)s}\sqrt{\et}\right)
ds\right)\\
&\qquad\cdot \exp\Big[C\int_{T_0}^T\Big(e^{-\delta_E\cdot s}\et+\br 4(f)+e^{-(1-\delta_{\mcr E}/2)s}\sqrt{\et}\Big)ds \Big]\\
&\leq\left(\varepsilon_0+C\sqrt{\varepsilon_0}(T-T_0)\right)\cdot \exp\Big[C\sqrt{\varepsilon_0}(T-T_0) \Big]\\
&\leq C\sqrt{\varepsilon_0}\exp(C\sqrt{\varepsilon_0}(T-T_0)).
}}
This implies
\eq{\alg{
\overline C\ab{\tau}\mcr G&\leq \overline C C\sqrt{\varepsilon_0}\exp(C\sqrt{\varepsilon_0}(T-T_0)-T)\\
&<\varepsilon_{\mathrm{tot}}/3,
}}
by choosing $\varepsilon_0$ sufficiently small. In total, we have shown that for sufficiently small $\varepsilon_0$ the estimates \eqref{small-1} and \eqref{small-2} hold with strict inequalities on $(T_0,T_+)$ and 
\eq{
\et\Big|_{T} \lesssim \et\Big|_{T_0}\exp\Big(-(1-\boldsymbol{\varepsilon}_{\mathrm{decay}})(T-T_0)\Big).
}
In particular, all relevant norms remain sufficiently small as $T\rightarrow T_+$ and by the continuation criterion the solution can be extended to $(T_0,T_++\varepsilon)$ for a small $\varepsilon$, where \eqref{small-1} and \eqref{small-2} hold on this extended interval. A standard continuity argument then implies $T_+=\infty$.

\subsection{Decay and asymptotic stability}
From the decay of the total energy the decay rates of the individual quantities can be inferred to read 

\eq{\label{final decay rates}\alg{
\Abk{g-\gamma}6&\lesssim \sqrt{\varepsilon_0} \exp\left[\left(-1+\frac{\boldsymbol{\varepsilon}_{\mathrm{decay}}-\delta_E}{2}\right) T\right]\\
\Abk{\Si}5&\lesssim\sqrt{\varepsilon_0} \exp\left[\left(-1+\frac{\boldsymbol{\varepsilon}_{\mathrm{decay}}-\delta_E}{2}\right) T\right]\\
\Abk{N-3}{6}&\lesssim\sqrt{\varepsilon_0} \exp(-T)\\
\Abk{X}6&\lesssim \sqrt{\varepsilon_0} \exp(-T)\\
\ev 54(f)&\lesssim\sqrt{\varepsilon_0}\exp\left[\left(\frac{\delta_{\mcr E}-(1-\boldsymbol{\varepsilon}_{\mathrm{decay}})}{2}\right) T\right]\\
\br 4(f)&\lesssim \sqrt{\varepsilon_0},
}}
where we recall $\boldsymbol{\varepsilon}_{\mathrm{decay}}>\delta_E$.
As an immediate consequence of these estimates, the rescaled metric converges against the Einstein metric $\gamma$ while all other terms decay. In total, the geometry converges in the above norms against the Milne geometry as $T\rightarrow\infty$.
\subsection{Future completeness}
For future completeness the rate of decay of the perturbation of the unrescaled geometry matters. We use the completeness criterion by Choquet-Bruhat and Cotsakis in \cite{CC02}. Therefore we change to inverse-CMC time $t_{\mathrm{icmc}}=-\tau^{-1}$, in particular $d\tau=\tau^2 dt$. The corresponding lapse and shift are related to the unrescaled variables and rescaled variables via $N_{\mathrm{icmc}}=\tau^2\widetilde N=N$ and $X_{\mathrm{icmc}}=\tau^2 \widetilde X=\tau X$. The metric and second fundamental form do not obtain additional factors of the mean curvature and we remain with $(\widetilde g,\widetilde \Si)=(\tau^{-2}g,\tau^{-1}\Si)$. 
Theorem 3.2 and Corollary 3.3 from \cite{CC02} provide as sufficient conditions for timelike and null geodesic completeness. Those are given and verified in the following. $(i)$ pointwise boundedness of the lapse $0<N_m<N_{\mathrm{icmc}}(t)<N_M$, which follows immediately from the pointwise estimate for the lapse.
$(ii)$ uniform boundedness for the metric $\widetilde g$ from below by some fixed metric for which we choose $t_0^2\gamma$. $(iii)$ Uniform boundedness of the shift vector, $\ab{X_{\mathrm{icmc}}}_{\tilde g}\lesssim \abg{X}<\sqrt{\varepsilon_0} t^{-1}$ follows from the decay estimates. Finally, we need to assure integrability of $(iv)$ $\ab{\nabla N_{\mathrm{icmc}}}_{\tilde g}=\ab{\tau}\abg{\nabla N}\lesssim \sqrt\varepsilon_0t^{-2}$ and $(v)$ $\ab{\widetilde{\Si}}_{\tilde g}=\ab{\tau} \abg{\Si}\lesssim \sqrt{\varepsilon_0} t^{-2+\varepsilon}$ on the interval $t\in(t_0,\infty)$. The decay rates in terms of the time $t$ immediately imply $(iv)$ and $(v)$. This proves the future completeness by the Corollary 3.3. from \cite{CC02} and finishes the proof of Theorem \ref{thm-1}.


\section*{Appendix}
\subsection*{A Formulae -- metric}
We collect several formulae here which are used in the course of the previous computations.

\begin{align}
\p T g^{ab}&=-g^{ac}g^{bd}\p Tg_{cd}\\
\p T\Gamma^i_{jk}&= \nabla_j(N\Si_k^i)+\nabla_k(N\Si^i_j)-\nabla^i(N\Si_{jk})\\
&\quad\nonumber\,+\nabla_j\widehat N\delta_k^i+\nabla_k\widehat N\delta_j^i-\nabla^i\widehat Ng_{jk}-\nabla_j\nabla^{(l}X^{i)}g_{kl}\\
[\p T,\Delta]X^i&=(\p Tg^{ab})\nabla_a\nabla_b X^i\\
&\quad\,\nonumber+g^{ab}\left( \na_a(\dot{\Gamma}^i_{jb}X^j)-\dot{\Gamma}^k_{ab}(\nabla_kX^i)+\dot{\Gamma}^i_{ja}(\na_bX^j)\right)
\end{align}
Also relevant for time differentiation of energies is the following formula for the time derivative of the Christoffel symbols.
\eq{
\hp 0 \Gamma^c_{ab}=-\nabla^c(Nk_{ab})+\nabla_a(Nk_{b}^c)+\nabla_b(Nk_a^c)
}

\subsection*{B Formulae -- matter} 
The divergence identity of the energy momentum tensor in the unrescaled form, $\widetilde{\nabla}_{\al}\widetilde T^{\al\be}$ reads in unrescaled variables (cf.~\cite{Re08}, (2.66), (2.67))
\eq{\alg{
\p t\tilde \rho-\tilde X^a\tilde{\nabla}_a \tilde\rho-\tilde N\tau\tilde\rho+\tilde N^{-1}\nabla_a(\tilde N^2\tilde j^a)-\tilde N\tilde k_{ab} \tilde T^{ab}=0\\
\p t\tilde j^b-\tilde X^a\nabla_a\tilde \jmath^b-\tilde X^b\nabla_a\tilde \jmath^a-\tilde N\tau\tilde\jmath^b+\nabla_a(\tilde N\tilde T^{ab})-2\tilde N\tilde k_a^b\tilde \jmath^a+\tilde \rho\tilde \na^b\tilde N=0.
}}
With respect to the rescaled variables, $\rho=\tilde{\rho}\ab\tau^{-3}$ and $\jmath=\ab\tau^{-5}\tilde \jmath$ these identities read
\eq{\label{tdrhoeta}\alg{
\p T\rho&=(3-N) \rho-X^a\na_a\rho+\tau N^{-1}\na_a(N^2\jmath^a)-\tau^2\frac N3g_{ab}T^{ab}-\tau^2N\Si_{ab}T^{ab}\\
\p T\jmath^a&=\frac53(3-N)\jmath^a-X^b\na_b\jmath^a-(\na^aX_b)\jmath^b+\tau\na_b(N T^{ab})-2N\Si_b^a\jmath^b-\ab\tau^{-1}\rho \na^aN
}}

\subsection*{C Time derivatives -- momentum functions}
\eq{\alg{
\p T\hat p&=\frac{1}{2\hat p}\Big[2\tau^2\langle \hat X,p\rangle_g^2+2\tau^2\langle\hat X,p\rangle_g\left(\langle\hat X,p\rangle_{\dot g}+\frac{1}{N}\langle p,\p TX-\hat X\p T N\rangle_g\right)\\
&\qquad\quad-(1+\tau^2\abg{p}^2)
\left(\ab{\hat X}^2_{\dot g}+\frac2N\langle \hat X,\p TX-\hat X\p T N\rangle_g\right)\\
&\qquad\quad+\tau^2(1-\abg{\hat X}^2)\left(2\abg p^2+p^ap^b\dot{g}_{ab}\right)\Big]
}}

\eq{\alg{
\p Tp^0&=2p^0\\
&\quad\,+\frac1{2N\hat p}\Big[ 4(p^0)^2(-N^2+\abg X^2)+6p^0\tau\langle p,X\rangle_g+2\tau^2\abg p^2+(p^0)^2\p T(-N^2+\abg X^2)\\
&\quad\qquad\qquad\,+2\tau p^0\langle p,\p TX\rangle+\tau^2p^ap^b\p Tg_{ab}\Big]
}}
by \cite{SaZa14}
\eq{\alg{
\p T\abg{p+\tau^{-1}p^0X}^2&=\ab{p+\tau^{-1}p^0X}_{\dot g}^2+4\abg p^2+4\tau^{-1}\langle p^0X,p\rangle_g\\
&\quad\,+2\tau^{-1}\langle p+\tau^{-1}p^0X,p^0X\rangle_g\\
&\quad\,+2\tau^{-1}\langle p+\tau^{-1}p^0X,(\p Tp^0)X+p^0\p TX\rangle_g
}}
\subsection*{D Momentum derivatives}
\eq{
\B_ep^0=\frac{\tau \hat X_ep^0+\tau^{2}N^{-1}p_e}{\hat p}
}
\eq{\alg{
\B_e\left(\frac{\abg{p+\tau^{-1}p^0X}^2}{\hat p}\right)&=\frac{2}{\hat p}(p_e+\tau^{-1}X_ep^0)\left(1+\frac\tau N\frac{\langle p,X\rangle_g}{\hat p}+\frac{p^0}{N}\abg X^2\right)\\
&\quad\,-\tau^2\frac{\abg{p+\tau^{-1}p^0X}^2}{\hat p^3}\left(\langle \hat X,p\rangle_g\hat X_e+(1-\abg{\hat X}^2)p_e\right)
}}

\subsection*{D2 Curvature of the tangent bundle}
\subsubsection{Curvature of the tangent bundle}
The connection coefficients $\boldsymbol{\Gamma}$ of the Sasaki metric $\mbf g$ with respect to the connection basis $\A_a=D_{x^a}, \B_a=\pp a$  take the following form.

\begin{equation}\label{co-coeff}
\alg{
\Chrn abc&=\Chr abc& \Chrn Ibc&=\tfrac12p^k\riem^{I-3}_{\quad kbc}\\
\Chrn aIc&=\tfrac12p^k \riem_{I-3,kc}^{\qquad\,\,\,\, a}&\Chrn abI&=\tfrac12p^k\riem_{I-3,kb}^{\qquad\,\,\,\, a}\\
\Chrn IJc&=\Chr {I-3}{J-3}{c}&\Chrn IbJ&=\Chrn aIJ=\Chrn IJK =0
}
\end{equation}
Here, we use small letters to denote indices in $\{1,2,3\}$ and capital indices to denote letters in $\{4,5,6\}$.
\subsection*{E Time derivative of the pressure}

\eq{\label{sdfjdsaa}\alg{
\p T\underline \eta&=\tau N\int \frac{p^a\A_a f}{Np^0}\frac{\abg{p+\tau^{-1}p^0X}^2}{\hat p}\sg dp\\
&\qquad\,+\p TX^a\Bigg[\int-\tau f\B_a\left(p^0\frac{\abg{p+\tau^{-1}p^0X}^2}{\hat p}\right)+2\tau^{-1}f\frac{p^0}{\hat p}(p_a+\tau^{-1}p^0X_a)\sg dp\\
&\qquad\qquad\qquad+\int 2\tau^{-1} f\frac{p^0}{2N\hat p^2}\langle p+\tau^{-1}p^0X,X\rangle_g\left(2p^0X_a+2\tau p_a \right)\sg dp\\
&\qquad\qquad\qquad-\int f\frac{\abg{p+\tau^{-1}p^0X}^2}{N\hat p^3}\left(\tau^2\langle\hat X,p\rangle_g p_a-(1+\tau^2\abg p^2)\hat X_a\right)\sg dp\Bigg]\\
&\qquad\,+\p TN\Bigg[\int  \tau^{-1} f\hat X^e\B_e\left(p^0\frac{\abg{p+\tau^{-1}p^0X}^2}{\hat p}\right)-2\tau^{-1}f\frac{(p^0)^2}{\hat p^2}\langle p+\tau^{-1}p^0X,X\rangle_g\sg dp\\
&\qquad\qquad\qquad+\int f \frac{\abg{p+\tau^{-1}p^0X}^2}{N\hat p^3}\left(\tau^2\langle\hat X,p\rangle_g^2-(1+\tau^2\abg p^2)\abg{\hat X}^2\right)\sg dp\Bigg]\\
&\qquad+\int f \frac1{\hat p}\left(\ab{p+\tau^{-1}p^0}_{\dot g}^2+4\tau^{-1}p^0\langle X,p\rangle_g+2\tau^{-1}\langle p+\tau^{-1}p^0X,X\rangle_g p^0\right)\sg dp\\
&\qquad+(\Si_{ab}+\frac{g_{ab}}3)X^e\int f\B_e\left(\frac{p^ap^b}{Np^0}\frac{\abg{p+\tau^{-1}p^0X}^2}{\hat p}\right)\sg dp\\
&\qquad+2\Gamma_e^e\int f \frac{\abg{p+\tau^{-1}p^0X}^2}{\hat p}\sg dp+2\Gamma^e_u\int fp^u\B_e\left(\frac{\abg{p+\tau^{-1}p^0X}^2}{\hat p}\right)\sg dp\\
&\qquad+\tau^{-1}{\overset{\circ}{\Gamma}}^e \int f \B_e\left(p^0\frac{\abg{p+\tau^{-1}p^0X}^2}{\hat p}\right)\sg dp\\
&\qquad-2\int f\left[2\frac{\abg p^2}{\hat p}\left(\frac{\tau}{N}\frac{\langle p,X\rangle_g}{\hat p}+\frac{p^0}{N}\abg X^2\right)+\frac{2}{\hat p}\tau^{-1}\langle p, X\rangle_g p^0\left(1+\frac{\tau}{N}\frac{\langle p,X\rangle_g}{\hat p}+\frac{p^0}{N}\abg X^2\right)\right]\sg dp\\
&\qquad+2\tau^2\int f \left(\frac{\abg{p+\tau^{-1}p^0X}^2}{\hat p^3}\right)\left(\langle \hat X,p\rangle_g^2+(1-\abg{\hat X}^2)\abg p^2\right)\sg dp\\
&\qquad+4\tau^{-1}\int \frac{f}{\hat p}\langle p+\tau^{-1}p^0 X,X\rangle_g\left(p^0+\frac1{2N\hat p}\left(2(p^0)^2(-N^2+\abg X^2)+3\tau\langle p,X\rangle_g+\tau^2\abg p^2+\frac{\tau^2}{2}\ab p^2_{\dot g}\right)\right)\sg dp\\
&\qquad-\int f\frac{\abg{p+\tau^{-1}p^0X}^2}{2\hat p^3}\Big(2\tau^2\langle\hat X,p\rangle_g^2+2\tau^2\langle \hat X,p\rangle_g\langle \hat X,p\rangle_{\dot g}\\
&\qquad\qquad\qquad\qquad\qquad\qquad-(1+\tau^2\abg p^2)\ab{\hat X}^2_{\dot g}+\tau^2(1-\abg{\hat X}^2)(2\abg p^2+\ab p^2_{\dot g})\Big)\sg dp\\
&\qquad+\int f \frac{\abg{p+\tau^{-1}p^0X}^2}{\hat p} g^{ab}\dot g_{ab}\sg dp+2\tau^{-1}\int f\frac1{2N\hat p^2}\langle p+\tau^{-1}p^0X,X\rangle_g(p^0)^2\ab{X}^2_{\dot g}\sg dp
}}

\end{document}